\documentclass[10pt,a4paper]{article}
\usepackage{amsmath}  
\usepackage{amsfonts}
\usepackage{amssymb}
\usepackage{amsthm}
\usepackage{cite}
\usepackage{fullpage}
\usepackage[]{nicefrac}
\usepackage{caption}
\usepackage{authblk} 
\usepackage{color}
\usepackage{float}
\usepackage[normalem]{ulem}
\usepackage{cancel}
\usepackage{graphicx}
\usepackage{varwidth}
\usepackage{tikz,tkz-graph}
\usepackage [resetlabels]{multibib}
\newcites{methods}{References}
\newcites{SI}{Supplementary References}
\usepackage{titlesec}
\usepackage{xstring} % conditional string
\usepackage{setspace} %spacing between lines
\usepackage{bbm}

\usepackage{tikz}
\usetikzlibrary{arrows}%, positioning} 
\tikzset{
	node distance=40pt,
	open tensor/.style={rectangle,draw,align=center,minimum size=25pt},
	fixed tensor/.style={fill,circle,inner sep=0pt,outer sep=0pt,minimum size=5pt}
}

\newcommand{\ignore}[1]{}
\newcommand{\ket}[1]{\left|#1\right\rangle}
\newcommand{\bra}[1]{\left\langle#1\right|}
\newcommand{\braket}[2]{\left\langle #1| #2 \right\rangle}
\newcommand{\ketbra}[2]{ \left| #1 \right\rangle\left\langle #2 \right|}

\theoremstyle{definition}
\newtheorem{theorem}{Theorem}
\newtheorem{corollary}{Corollary}%[theorem]
\newtheorem{definition}{Definition}
\newtheorem{lemma}{Lemma}
\newtheorem{claim}{Claim}

\newtheorem{remark}{Remark}
\newtheorem{fact}{Fact}
\begin{document}
	
	\title{Quantum Circuit-depth Lower Bounds For Homological Codes} 
	
	\author{Dorit Aharonov ~\& Yonathan Touati} 
	\affil{School of Computer Science and Engineering, 
		The Hebrew University, \\ Jerusalem,  Israel.}
	% & the Quantum Information Science Center}}
	
	\maketitle
	
	\begin{abstract}
		We provide an $\Omega(log(n))$ lower bound for 
		the depth of any quantum circuit generating the unique groundstate of 
		Kitaev's spherical code. No circuit-depth lower bound was known 
		before on this code in the general case where the gates can connect 
		qubits even if they are far away; 
		It is a known hurdle in computional complexity to handle general circuits, 
		and indeed the proof requires introducing 
		new techniques beyond those used to prove the $\Omega(\sqrt{n})$ 
		lower bound which holds in the geometrical case \cite{TQO}. 
		The lower bound is tight (up to constants) since a MERA circuit 
		of logarithmic depth exists\cite{MERA}.  
		To the best of our knowledge, this is the first time a quantum 
		circuit-depth lower bound is given for a unique 
		ground state of a {\it gapped} local Hamiltonian. Providing a lower bound in 
		this case seems more challenging, since such 
		systems exhibit exponential 
		decay of correlations \cite{hastingskoma} and standard lower bound techniques 
		\cite{aharonovNisanKitaev} do not apply. 
		We prove our lower bound by introducing the new
		notion of $\gamma$-separation, and analyzing 
		its behavior using algebraic topology arguments.

		We extend out methods also to a wide class of polygonal complexes 
		beyond the sphere, and prove a circuit-depth lower bound whenever the 
		complex does not have a small "bottle neck" (in a sense which we define). 
		Here our lower bound on the circuit depth is only $\Omega(logloglog(n))$. 
		We conjecture that 
		the correct lower bound is at least $\Omega(loglog(n))$, but this 
		seems harder to achieve due to the possibility of hyperbolic geometry. 
		For general simplicial complexes the lack of geometrical 
		restriction on the gates becomes considerably more problematic than 
		for the sphere, and we need to thoroughly modify the original argument in 
		order to get a meaningful bound. 
		
		To the best of our knowledge, this is the first time the class of 
		trivial quantum states is separated from the class of unique groundstates of 
		gapped local Hamiltonians, improving our understanding of 
		the heirarchy of global entanglement and topological order; 
		we provide a survey of the current status of this heirarchy for completeness. 
		We hope the tools developed here will be useful in various contexts in which 
		quantum circuit depth lower bounds are of interest, 
		including the study 
		of topological order, quantum computational complexity and 
		quantum algorithmic speed-ups. 	
	\end{abstract}

	\section{Introduction}
	Since the early days of quantum mechanics, physicists have tried to 
	understand and quantify entanglement in quantum systems. 
	Whereas quantifying two-body entanglement is rather well understood, using 
	Von Neuman entropy of subsytems, the question is far more involved 
	when many-body entanglement is concerned. 
	In \cite{Hastingstrivial} 
	Hastings suggested a very interesting definition which captures 
	perhaps the first step in understanding this question: 
	He suggested to use, as a first order approximation of global 
	entanglement or {\it topological order}, 
	the requirement that a state be {\it non-trivial}, 
	where a quantum state is trivial if it can be 
	generated (or approximated) efficiently by a constant depth quantum circuit: 
	
	\begin{definition} {\bf Trivial pure states (roughly)} 
		\label{def:trivialStates}
		A family $\{\ket{\psi_n}\}$ of pure states on n qubits is called trivial if for every $n$, $\ket{\psi_n}=U_n\ket{0}^n$ where $U_n$ is a quantum circuit of bounded depth $d_n=O(1)$ made of quantum gates acting on at most $O(1)$ 
		qubits each. 
	\end{definition}
	
	One can include a notion of approximation in a natural way \cite{Hastingstrivial}. 
	If one is interested in simulating local observable on such states, 
	as is often the case in physics as well as in quantum complexity, then the term
	"trivial" is well justified. 
	This is because sampling local 
	observables of such a state $\ket{\psi}$ can be done efficiently with a classical computer, given the classical 
	description of the shallow quantum circuit that generated it.  The way to do this is to notice that the size of the {\it light cone} of each output particle, namely the number 
	of qubits in the input and throughout the evolution affecting its result, is $O(1)$. 
	From a more physics-like point of view, 
	trivial quantum states are those states that are adiabatically connected, via a 
	path of gapped Hamiltonians, to a tensor product state\cite{ChenGuWen}; 
	this is the simplest type of a 
	quantum states from the point of view of entanglement, and so
	this means that there is a constant time adiabatic 
	evolution which generates this state starting from a tensor product state. 
	In Physics language, this adiabatic connection without closing the gap 
	is a way to partition the set of groundstates to {\it equivalence classes}, 
	called ``phases''; and so trivial states are viewed as belonging to the 
	same ``phase'' as tensor product states, constituting
	the simplest possible states from the point of view 
	of entanglement.  In some sense, all other 
	states are viewed as having {\it topological order}
	(e.g., \cite{FreedmanHastings}). 
	
	We note that surprisingly, even such simple 
	quantum states as trivial states can exhibit extremely interesting 
	correlations, when considering measuring {\it all} 
	of their particles 
	\cite{DivincenzoTerhal, BravyiKonig, IQC}; The correlations that can 
	be generated 
	even in such simple quantum evolutions can be computationally hard classically.
	Yet, in the most common situation in quantum complexity, 
	in which local measurements are of interest, 
	states in this class can indeed be regarded as
	``simple'' from the point of view 
	of entanglement. 
	
	We take this point of view, and ask 
	whether a state (rather, a family of states) 
	is trivial or not. 
	From a physics point of view, 
	the fact that a state is trivial provides 
	strong intuition about how limited its global entanglement is, 
	and is a basic question in the study of topological order 
	\cite{Hastingstrivial, FreedmanHastings}. 
	The question is of importance of course 
	also in quantum computation theory (see \cite{B}): 
	proving non-triviality of states means proving 
	a {\it lower bound} on the minimal-depth of the circuit generating the state, 
	and such lower bounds are a central topic of interest when understanding 
	the complexity of quantum states. 
	As a notable example, a study of Freedman and Hastings \cite{FreedmanHastings} 
	connects between the triviality of states 
	and the major open problem of resolving the 
	quantum PCP conjecture\cite{Detectabilitylemma,QPCP}. 
	One version of this conjecture, called the 
	{\it gapped} qPCP version, states that 
	deciding whether the ground energy of a local Hamiltonian with $M$
	local terms is $0$ or at least a constant fraction of $M$, is QMA hard.   
	Freedman and Hastings noted that assuming $qPCP$ holds 
	(and assuming the class $QMA$ not equal to $NP$), 
	one needs to at least be able to 
	point at a family of local Hamiltonians whose low energy states 
	are {\it all} non-trivial; 
	the existence of such  family is called the 
	NLTS conjecture \cite{FreedmanHastings}
	(and it is still wide open, despite some recent progress \cite{EldarHarrow}). 
	Thus, whether or not a family of quantum states is trivial, is of interest 
	in both complexity and physical context.   
	
	{~}
	
	\noindent{\bf Simple arguments for quantum circuit depth lower bounds} 
	There are very few known techniques for 
	proving circuit lower bounds
	of quantum states, and very few classes of non-trivial states known. 
	The easiest argument in this direction 
	was given already in \cite{aharonovNisanKitaev} (see also \cite{Chen2})  
	to show that the $\ket{CAT^+}$ 
	state $\frac{1}{\sqrt{2}}(|0^n\rangle +|1^n\rangle)$ is non-trivial 
	(albeit not in 
	this terminology), and requires $\Omega(log(n))$ circuit-depth.  
	The argument goes by noting that any two qubits 
	in $\ket{CAT^+}$  are correlated (i.e., their reduced density matrix is far 
	from a tensor product); while a 
	a constant depth circuit correlates any given qubit with at 
	most a {\it constant} 
	number of other qubits. 
	This simple correlation based argument, however, is rather limited, since 
	states in which any two qubits are correlated are rather special. 
	When correlations are more complicated (in particular, multipartite) 
	this argument will not work. 
	
	We note that the above correlation-based argument can 
	be adapted to provide an example of 
	a non-trivial state which is also a unique groundstate 
	of a local-Hamiltonian; while it is easy to see that $\ket{CAT^+}$  itself
	is not the unique 
	groundstate of any local Hamiltonian (since it agrees locally with 
	$\ket{CAT^-}$ ), 
	using the by-now-standard method of the circuit-to-Hamiltonian 
	construction \cite{KitaevQMA}(as suggested to us by Itai Arad \cite{arad})
	one can construct a state whose restriction to a non-negligible 
	subsystem is very close to $\ket{CAT^+}$ e, and yet 
	the overall state 
	is a unique groundstate of a local Hamiltonian (See Appendix B). 
	
	Another rather simple argument implying non-triviality, 
	works for quantum error correcting code 
	states. The argument uses the following simple fact:  
	\begin{fact}\label{fact:trivialareunique}
		{\bf trivial quantum states are unique ground states of 
			gapped local Hamiltonians}
		Let $\ket{\psi_n}$ be the output state of a constant depth circuit on the tensor state $\ket{0}^n$, then $\ket{\psi_n}$ is the unique ground state of a gapped local Hamiltonian $H=\sum_{i=1}^{m}H_i$. 
	\end{fact}
	We will give a formal proof of this fact in Appendix A,
	however it would be beneficial to 
	explain the argument here since its underlying
	idea constitutes a starting point for many results in our paper. 
	The argument starts by noting that by definition, a trivial state 
	can be written as $U\ket{0^n}$ for some constant depth $U$.  
	The idea is to notice that $\ket{0^n}$ is the 
	unique groundstate of the gapped local Hamiltonian 
	$H=\sum_i\ket{1}\bra{1}_i$ which projects each of the qubits on its $\ket{1}$ state; 
	to generate a local Hamiltonian whose unique groundstate is $U\ket{0^n}$
	we need to change the basis of $H$ by $U$, namely consider terms of the form   
	$U(\ket{1}\bra{1}_i\otimes I_{[n]\backslash \{i\}})U^\dagger$. The key point is that 
	this term acts non trivially only on a constant number of qubits, which 
	are within the {\it light cone} on the $i$'th qubit, since all gates in $U$
	outside of the light cone of the $i$th qubit, cancel with their corresponding 
	inverse gate in $U^{-1}$. For a detailed proof see Appendix A.  
	
	Let us now see how to deduce from this easy fact,
	that any state in any quantum error correcting code whose distance $d$ is 
	more than a constant, is non-trivial. 
	Recall that any two states in such a QECC 
	must have the same reduced density matrix on any set of $k$ qubits,
	for $k< d$ (See Fact \ref{fact:qecc}). 
	Hence, if one of the states in a QECC of distance $d>k$ 
	is a groundstate of a $k-$local Hamiltonian, 
	so are all the others, and hence none of them can be the 
	{\it unique} groundstate of that $k$-local Hamiltonian.
	By Fact \ref{fact:trivialareunique}, such a state cannot be a trivial state.  
	This argument was essentially the same as the argument 
	given in \cite{B}, except that  
	they use the Lieb-Robinson bound (which is the physics continuous analog of 
	the light cone notion), and this  
	makes the argument somewhat less transparent for computer scientists. 
	
	To the best of our knowledge, these two simple 
	arguments are the only currently known arguments 
	for proving non-triviality (or more generall, circuit-depth lower bounds) 
	of groundstates of local Hamiltonians. 
	
	{~}
	
	\noindent{\bf Geometrical non-triviality} 
	While proving general non-triviality of groundstates in general seems 
	rather difficult, the task becomes more accessible when 
	we add a geometrical 
	restriction on the gates of the quantum circuit $U$. 
	Namely, we will consider \emph{geometric trivial} states, defined in the same 
	fashion as triviality of states except for
	considering only quantum ciruits $U$ whose 
	gates act on qubits lying within a constant distance from each other, where    
	the metric is given by the interaction graph of the Hamiltonian (i.e., 
	two qubits are 
	connected if there is a term in the Hamiltonian acting on them).
	
	\begin{definition}{\bf Geometrically Trivial pure states} \label{def:GeomTrivSt}:  
		A family $\{\ket{\psi_n}\}$ of pure states defined on $n$ qubits sitting 
		on the edges of a graph is called {\it geometrically trivial} 
		if for every $n$, $\ket{\psi_n}=U_n\ket{0}^n$ where $U_n$ is a quantum circuit of bounded depth $d_n=O(1)$ made of quantum gates acting on at most 
		$O(1)$ qubits each, where all of the qubits acted upon by the same gate 
		are within $O(1)$ distance from each other in the natural graph metric.
	\end{definition}
	
	Under this geometrically restricted circuit setting, 
	Bravyi \cite{BravyiComment} provided a beautiful argument 
	that the toric code states cannot be generated by constant 
	depth quantum circuit. This indeed follows 
	already from the fact that the toric code is a QECC of distance
	$\Omega(\sqrt{n})$, 
	by the argument sketched above; 
	however, Bravyi's argument is very different, 
	and relies on a beautiful topological consideration.
	
	Importantly for the focus of this paper, 
	Bravyi's argument (which was never published and
	so we provide it here for completeness; see Subsection \ref{sec:geocube}) 
	can be applied also in the case of Kitaev's code on the 
	{\it sphere}, which is defined similarly to the toric code but has dimension 1. For that reason, 
	the above argument relying on 
	large distance QECC
	does not hold whereas Bravyi's argument does and directly implies {\it geometrical} non-triviality of the spherical code state.   
	Freedman and Hastings \cite{FreedmanHastings} also used topological 
	tools to argue geometrical non-triviality of certain other quantum 
	groundstates. 
	
	The above proofs do work for {\it unique} ground states of {\it gapped} 
	local Hamiltonians, but importantly,
	they are restricted to {\it geometrical} non-triviality 
	(and in fact, clearly derive lower bounds which are too strong to apply 
	when the geometrical restriction on the gates 
	is relaxed).  
	
	{~} 
	
	\noindent{\bf Moving towards {\it general} non-triviality of unique groundstates of 
		gapped Hamiltonians}  
	In light of all the above, one might ask:  
	Could it be true that all unique groundstates of gapped local Hamiltonians 
	are trivial, if one allows gates which are non-geometrically restricted?  
	In other words, could it be that Fact \ref{fact:trivialareunique} 
	is "if and only if"? 
	Hastings and Koma \cite{hastingskoma} proved that ground states of 
	gapped Hamiltonians 
	exhibit a phenomenon of decay of correlations; this suggests that the 
	entanglement in such groundstates is severely limited, and further 
	suggests that the possibility that 
	such states are always trivial, at least cannot be ruled out immediately, since the argument sketched above of correlations implying non-triviality cannot be applied. Moreover, the argument of non-triviality for quantum error correcting code states cannot be applied either in this case. 
	
	In this paper we would indeed like to prove general circuit lower bounds
	in such cases.   
	For a start, one would like to be able to prove a circuit-depth lower bound 
	for the celebrated state of Kitaev's spherical code, in the non-geometrical setting. 
	
	The task of providing general lower bounds is notoriously hard in computational 
	complexity; it seems that general methods for providing such lower bounds do 
	not 
	exist. Apart from the two simple arguments sketched above, which are suited for 
	non-geometrical lower bounds, we are aware of only one 
	technique for proving such quantum circuit depth lower bounds, given by Eldar 
	and Harrow\cite{EldarHarrow}. 
	Eldar and Harrow provide non-geometrical circuit depth lower bounds 
	for a certain class of quantum states, using a very innovative 
	technique related to 
	expansion of probability distributions; however 
	their methods inherently do not apply for {\it unique} ground 
	states of local Hamiltonians, 
	since, just like in the CAT state and in the QECC codes arguments above, at the 
	bottom of their argument lies the existence of {\it two} states which 
	are globally orthogonal but locally similar; such an argument cannot 
	be used when the groundstate is unique.   
	
	\subsection{Results}
	In this paper, we prove quantum circuit depth lower bounds for a large class of 
	quantum states, for general 
	non-geometrically restricted circuits. In particular, we prove an 
	$\Omega(log(n))$ lower bound for Kitaev's 
	spherical code state; we then extend the results to 
	various groundstates of codes defined on polygonal complexes (as long as they don't have 
	small bottlenecks). To this end we develop various tools, 
	mostly borrowed from algebraic topology. 
	
	To our knowledge, this work gives the first known example of  
	non-trivial quantum states that are unique ground 
	states of gapped Hamiltonians, in which the decay of correlations 
	due to Hastings and Koma \cite{hastingskoma} holds. 
	
	{~} 
	
	\noindent{\bf Kitaev's spherical code.} 
	We start with Kitaev's surface codes. Instead of a surface code on a sphere, 
	we work with the code set on a large $2$-dimensional 
	cube whose $6$ faces are each tesselated to squares 
	using an $n$ by $n$ grid, and
	the corresponding local Hamiltonian is the usual star and plaquettes Hamiltonian. 
	It is well known that in this case the groundstate is unique. 
	Bravyi's \cite{B} topology-based proof of non-triviality of the toric code state, holds also 
	in this case. 
	
	\begin{theorem}(Adapted from \cite{BravyiComment})
		\label{thm:NonGTrivCubeState} 
		Let $\{\ket{\psi_i}\}_{i=1}^\infty$ be the family of cube states on $N_i$ qubits. Furthermore, let $U_i$ be a quantum circuit of depth $d_i$ using geometrically local gates on at most $c=O(1)$ qubits, such that $U_i\ket{0^{N_i}}=\ket{\psi_i}$. 
		Then, $d_i = \Omega(\sqrt{N_i})$.
	\end{theorem}
	
	However, this method assumes that the circuits have the same 
	geometry as the grid, and so this result only shows
	geometrical non-triviality (see definition \ref{def:GeomTrivSt}). 
	
	Our first result is to show that the state
	is non-trivial regardless of how far from each other
	the qubits on which the c-local gates act are. 
	In fact, we prove a logarithmic lower bound on any generating 
	quantum circuit of the cube states:
	
	\begin{theorem}
		\label{thm:nonTrivCubeState}
		Let $\{\ket{\psi_i}\}_{i=1}^\infty$ be the family of cube states on $N_i$ qubits, and for every $i\geq 1$, let $U_i$ be a quantum circuit of depth $d_i$ using gates of locality $c=O(1)$, such that $U_i\ket{0^{N_i}}=\ket{\psi_i}$. Then, $d_i=\Omega(log(N_i))$. 
	\end{theorem}
	
	We believe this result was commonly assumed to be true, however, a proof did not exist. In fact, the proof is far from being 
	straight forward; 
	The non-geometrical gates pose considerable complications, and the proof turns out to be non-trivial (see overview of proofs). 
	
	We next proceed to prove a similar result for a much more general class of states. 
	We note that on one hand, the cube states is a family of states which admits a simple geometric description - 
	and therefore is a good candidate to illustrate some of the ideas in our more general proof. On the other hand, 
	it also has the interesting property of yielding states which are unique ground 
	states of gapped local Hamiltonians, making it resistant to the two simple methods 
	sketched in the introduction, for proving lower bounds on general quantum circuit-depth; the cube states are subject to the exponential decay of correlations \cite{hastingskoma} and exhibit no obvious long range correlations; they are also  not members of any 
	QECC with non-constant distance.  
	
	{~}
	
	\noindent{\bf Extension to polygonal complexes} 
	We then generalize this result to a wide class of quantum states, 
	defined as codes on what we call "closed surface complexes" (CSCs), 
	where the code states are stabilized by the usual star and plaquette 
	operators on these complexes. 
	This is done under a
	topological condition on the complex, which we call
	$r$-{\it simply connectedness}.

	We start by showing that the geometrical non-triviality 
	topology based proof of \cite{B} extends  
	to geometrical non-triviality in the case of CSCs as well: 
	
	\begin{theorem}\label{thm:NonGTrivSStates}
		Let $\{C_i\}_{i=1}^\infty$ be a family of (possibly non degenerate) surface codes on $N_i$ qubits defined on the closed surface complexes ${G_i(V_i, E_i, F_i)}_{i=1}^\infty$ (see definition \ref{def:CSC}) such that $deg(G)=O(1)$, $\widehat{deg}(G)=O(1)$ and for all $i\geq 1$, let $\ket{\psi}_i\in C_i$. Furthermore, let $U_i$ be a quantum circuit of depth $d_i$ using geometrically local gates on at most $c=O(1)$ qubits, such that $U_i\ket{0}^{N_i}=\ket{\psi_i}$. 
		Assume that  $G_i$ is $f$-simply connected for some $f=O(log(N_i))$, then, $d_i = \Omega(log(N_i))$.
	\end{theorem}
	
	Our main result is a generalization of this to general 
	(non-geometrically restricted) quantum circuits: 
	
	\begin{theorem} [Surface states are non trivial]
		\label{thm:NonTrivSStates}
		Let $\{C_i\}_{i=1}^\infty$ be a family of (possibly non degenerate) surface codes on $N_i$ qubits defined on the  closed surface complexes ${G_i(V_i, E_i, F_i)}_{i=1}^\infty$ such that $deg(G)$, $\widehat{deg}(G)=O(1)$ and for all $i\geq 1$, let $\ket{\psi}_i\in C_i$. Furthermore, let $U_i$ be a quantum circuit of depth $d_i$ using gates acting on $c=O(1)$ qubits such that $U_i\ket{0}^{N_i}=\ket{\psi_i}$. Assume that  $G_i$ is $f$-simply connected, for some $f=O(log(log(N_i)))$ then, $d_i = \Omega(log(log(log(N_i))))$.
	\end{theorem}
	
	Here, $deg(G)$ and $\widehat{deg}(G)$ are the maximal vertex and face degree of G as defined in definition \ref{def:VertFaceDeg}.
	The condition on the face and vertex degrees is just a way to restate the fact that our states should be ground states of \emph{local} Hamiltonians only. 
	Broadly speaking, $r$-simply connectedness means that the subcomplex induced by any ball of radius at most $r$ should be 
	simply connected (see Definition \ref{def:AlphaContr}). 
	In other words, $G_i$ shouldn't have any "bottlenecks". Observe that 
	at first glance, this condition might seem
	weaker than simple connected of the whole complex; it also seems 
	as if it is implied by it. However this is not the case, as can be seen in 
	Figure \ref{fig:bottleneck}.
	This is quite subtle, and indeed, our proof does NOT work on some 
	spaces which are simply connected, but have small bottle necks so they are NOT 
	$r$-simply connected for some small $r$. 
	
	The main difference between Theorem \ref{thm:nonTrivCubeState} and Theorem \ref{thm:NonTrivSStates} lies in the different lower bounds they exhibit. In fact, this is a reflection of the fundamental difference between Euclidean and hyperbolic metric spaces, 
	and more precisely between the ratios of the area of disks to their radii 
	in both cases. We will further discuss this issue in section \ref{sec:NonTrivCubeState}.
	
	\subsection{Proofs overview}
	Our starting point for this paper is Bravyi's topological argument for $\Omega(\sqrt{n})$-circuit depth lower bound for toric code states \cite{BravyiComment}. This argument works essentially as is for Kitaev spherical code states. 
	For completeness, we rederive this argument fully in this paper, since it was not published; in fact, we provide the proof for a slightly different set of states, which we call {\it cube states}, defined in Section \ref{sec:geocube}. 
	
	Bravyi's proof relies on   
	a central lemma, Lemma \ref{lem:ComutOp}, which will be the basis point for all our proofs. This lemma is proved in Section \ref{sec:Bravyi}. 
	The lowerbound for cube state itself is then derived in Section 
	\ref{sec:geocube}. 
	
	We will then explain how starting from Bravyi's approach we can provide a non-geometrical lower bound, and later 
	we analyze both geometrical and non-geometrical lower bounds for the more general case of Closed simplicial complexes, CSCs (see definition \ref{def:CSC})
	
	\subsubsection{Bravi's argument}
	We now sketch the central lemma in Bravyi's argument. Given a grid representing a closed surface,
	suppose $U$ is a constant depth 
	circuit generating the groundstate of some Hamiltonian set on the grid.  
	We choose two distant edges $e$ and $f$ on the cube grid, and consider a path $\gamma$ between
	them. The effective support $B_{[\gamma],U}$ of the path $\gamma$,
	relative to a generating circuit $U$ is the following set of edges: 
	Consider the set of all qubits that are non-trivially acted upon, if we conjugate 
	the path $\gamma$ by the circuit $U$; if $U$ has constant depth, 
	this set takes the form of a "thickening" (of constant thickness) of $\gamma$, 
	[as in the argument behind the proof of fact \ref{fact:trivialareunique} (see Appendix A)].
	We look at the intersection of these subsets when going over all paths $\gamma$ connecting $e$ to $f$, 
	and we call this the {\it effective support} of $\gamma$ with respect to $U$, denoted 
	$B_{[\gamma],U}$. 
	This turns out to be the union of two small regions 
	around $e$ and $f$. 
	(See Figure \ref{fig:sausage}). 
	
	Lemma \ref{lem:ComutOp} states 
	that any operator which stabilizes the groundstate, and  
	whose support lies completely 
	outside of $B_{[\gamma],U}$ for the path $\gamma$, must 
	commute with any Pauli $Z$ operator supported on $\gamma$, when acting on the 
	groundspace. This relies on the fact that the circuit is constant depth, and thus under a suitable change of basis, all arguments about commutation relations between 
	operators on different sets can essentially be made ``classical'', 
	by considering the operators supports except ``thickened'' by some constant. We next use this lemma to derive the geometrical lower bound. 
	
	\subsubsection{Cube states: the geometric case}
	In section \ref{sec:geocube}, we prove geometric non-triviality of cube states, which 
	are easy-to-work-with variant of the spherical codes of Kitaev (see \cite{ToricCode}).  
	More specifically, We consider the quantum codes associated with regular tesselations 
	of the cube, with the usual definition of star and plaquette operators (see 
	Definition \ref{def:Stab} in Subsection \ref{Subsec:SurfaceCodes}). Since the cube is simply connected, 
	these codes have dimension $1$ (see corollary \ref{cor:homologyDimensionCSC}). 
	These unique ground states are the "cube states".

	In order to prove non triviality, we first show in Lemma \ref{lem:SausageIntersection} that for the cube states - assuming the generating circuits $U_n$ are geometrically local and that they have constant depth -
	the effective supports (see definition \ref{def:EffSupp}) of any path $\gamma$ is contained in the union of two
	balls of radius proportional to the depth of $U$ and centred around the 
	endpoints of $\gamma$. 
	This fact has a rather intuitive geometric explanation (see Figure \ref{fig:sausage}). 
	
	Finally, using Lemma \ref{lem:ComutOp}, one can 
	derive a constradiction: We choose
	both enpoints of $\gamma$ to be on opposite faces of the cube, and
	find a large {\it closed} copath $\beta$ far from these end points 
	(for now, think of this as a huge 
	loop around the cube) cutting the cube into  
	two connected components, each one containing one of the endpoints. This can be done such that \emph{any} path between the endpoints intersects 
	the copath an odd number of 
	times since it starts from one connected component and ends up in the second one (see Figure \ref{fig:cube}).  
	Now consider the following 2 operators: The first $\gamma_Z$
	applying $Z$'s on all the qubits in $\gamma$, 
	and the second, $\beta_X$ applying $X$'s on all the qubits in $\beta$. 
	Since $\gamma$ and $\beta$ intersect in an odd number 
	of edges, these global operators 
	anticommute. However note that $\beta_X$
	stabilizes the groundstate, and its support is 
	outside of 
	$B_{[\gamma],U}$ since it is far from the endpoints of $\gamma$. 
	This finally yields a contradiction to Lemma \ref{lem:ComutOp} 
	claiming those two operators should  
	commute.

	\begin{figure}
		\begin{center}
			\includegraphics[scale=0.5]{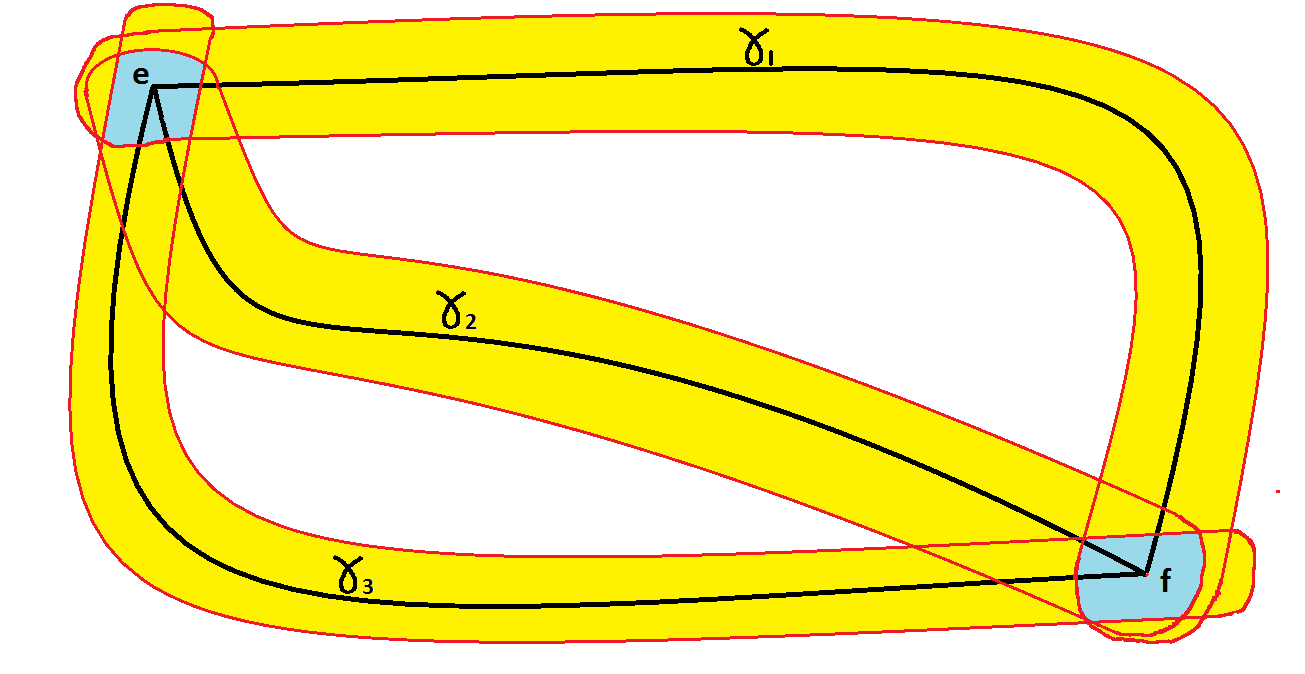}
		\end{center}
		\caption{\label{fig:sausage}The yellow areas represents the lower light cone of 3 distinct paths $\gamma_1$, $\gamma_2$  and $\gamma_3$  between e and f. The blue area centered around $e$ and $f$ represent the intersection between those 3 areas when going through these 3 paths. The set A of edges in the intersection of the yellow "fattenings" for all paths in $\Gamma_{e,f}$ is contained in the union of 2 balls centered around $e$ and $f$ and whose radii is upper bounded by some function of the circuit depth}
	\end{figure}
	
	\begin{figure}
		\begin{center}
			\includegraphics[scale=0.4]{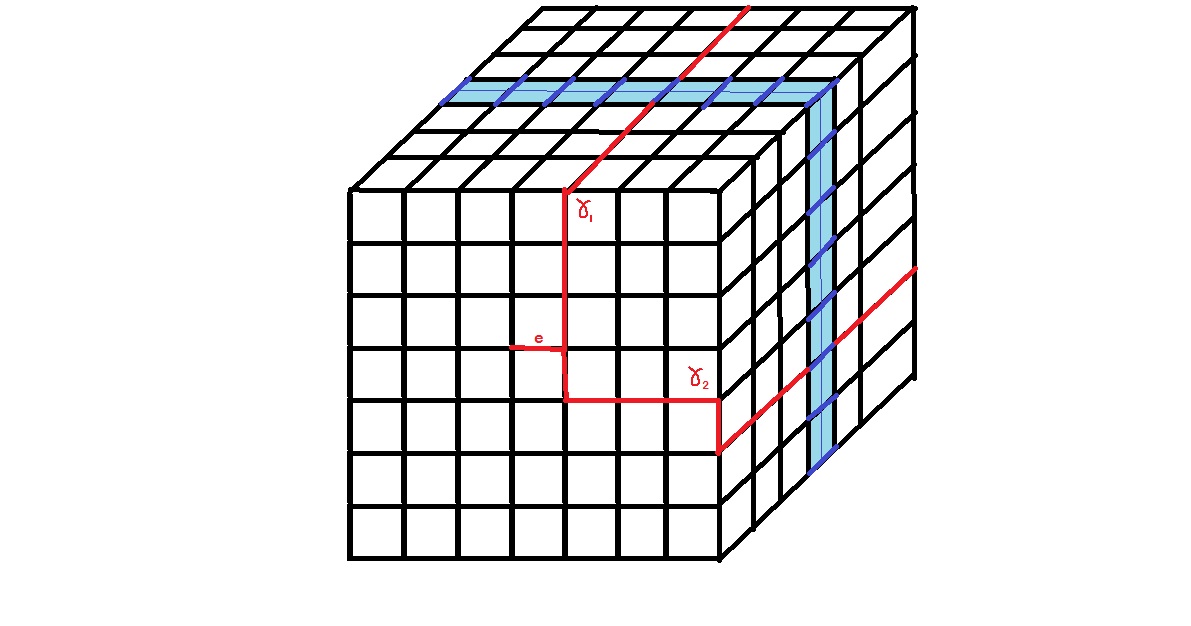} 
		\end{center}
		\caption{\label{fig:cube} The co-path in blue separates the edge e from f on the other side of the cube: Both $\gamma_1$ and $\gamma_2$ have to intersect it non-trivially}.
	\end{figure}

	\subsubsection{Cube states: the non geometric case}
	
	The key issue motivating this work is that the above arguments stop working when one drops the spatial locality assumption. Indeed, in that case, $B_{[\gamma],U}$ cannot be shown to be neither small, nor spatially close to the endpoints.
	In section \ref{sec:NonTrivCubeState}, we solve that issue for the special case of the 
	cube states, by showing that the very fact that $B_{[\gamma],U}$ could be large can actually be used to prove the non triviality of the cube state. In fact the first step is to prove in lemma \ref{lem:LargeB} that $|B_{[\gamma],U}|$ is actually lower bounded by the length of $\gamma$: Indeed, suppose otherwise, then both endpoints $e$ and $f$ of $\gamma$ would lie in a different connected component of $|B_{[\gamma],U}|$. One can now define an operator $\beta_X$ on the coboundary $\beta$ of the connected component $T_e$ containing $e$. Observe that $\gamma$ and $\beta$ must intersect on an odd amount of edges: indeed, we can consider $\gamma$ as an indexed finite sequence of vertices where each edge in $\gamma\cap\beta$ is interpreted as getting "in" or "out" of $T_e$. But since $e\in T_e$ and $f\notin T_e$, we immediately get $|\gamma\cap\beta|=0\; mod\; 2$. It follows that $\gamma_Z$ and $\beta_X$ must anticomute. On the other hand, from lemma \ref{lem:ComutOp}, since $\beta_X$ is supported on $E\backslash B$, we expect $\beta_X$ and $\gamma_Z$ to commute, leading to a contradiction.
	
	To this end, we introduce the notion of $\gamma$-separation: A set of edges is said to be $\gamma$-separating if any path $\gamma'\in [\gamma]$ intersects it non trivially. We show in claim \ref{cl:LCaGammaSep} that the upper light cone of any edge in A is $\gamma$-separating. Then, we proceed to prove in claim \ref{cl:GammaSepCon} and corollary \ref{cor:AisConnected} that W.l.o.g, those lightcones can be assumed to be connected. Indeed: given 2 different paths $\gamma$ and $\gamma'$ intersecting 2 distinct copath connected components, one can find a new path $\beta\in [\gamma]$ such that $\beta$ doesn't intersects with either component, contradicting the assumption of $\gamma-$separation.
	
	In lemma \ref{lem:blockage}, we show that those copath connected $\gamma$-separating lightcones $\{LC(a)|a\in A\}$ must lie within the union of 2 balls centered around e and f of radius $|LC(A)|$. This is rather obvious for the cube complex since one can always find a ball $C$ of radius $|LC(a)|$ and containing  $LC(a)$ such that both endpoints of $\gamma$ lie outside of $C$. Then, removing C from the original cube complex yields a connected subcomplex where any path between e and f doesn't intersect $LC(a)$ leading to a contradiction.
	
	Finally, we prove Theorem \ref{thm:nonTrivCubeState} by using the fact that on one hand, $A$ must be a large subset from lemma \ref{lem:LargeB} (namely, its size is bounded from above by the diameter of the cube), but on the other hand, $A$ is included inside the lower lightcone of the small balls around e and f defined earlier, whose radii depend of the depth of the generating circuits $U_n$. Therefore, we can extract a relation between the circuit depth and the diameter of the graph wich is known to be $O(n)$=$O(\sqrt N)$ for the cube where N is the number of qubits in the system. \\
	
	\subsubsection{From cube states to complexes: the geometric case}
	
	In section \ref{sec:alphaCont}, we develop the tools to prove Theorem \ref{thm:NonTrivSStates}, which generalizes the above result to more general closed surface complexes. The main issue with the previous approach comes from the fact that lemma \ref{lem:blockage} doesn't hold in the general setting. Indeed, one could construct a CSC with small "bottlenecks", such that any path between e and f must pass through it. Therefore, the upper light cone of some edge in A could lie in that small domain even though it could be far away from both e and f (see figure \ref{fig:bottleneck}). This example is the main motivation behind the additional assumption of $r$-simple connectedness we require in order to avoid such complications.
	Namely, a CSC is said to be $r$-simply connected if the first homology and cohomology groups of any ball of radius $r$ around any edge of the complex vanishes. This can be shown to be enough to ensure that Theorem \ref{thm:getAround} - a more general version of lemma \ref{lem:blockage} - holds in the framework of $r$-simply connected CSCs. 
	The proof of  Theorem \ref{thm:getAround} relies on rather technical observations encompassed in lemma \ref{lem:ConnectedBoundary}. \\

	\begin{figure}
		\begin{center}
			\includegraphics[scale=0.25]{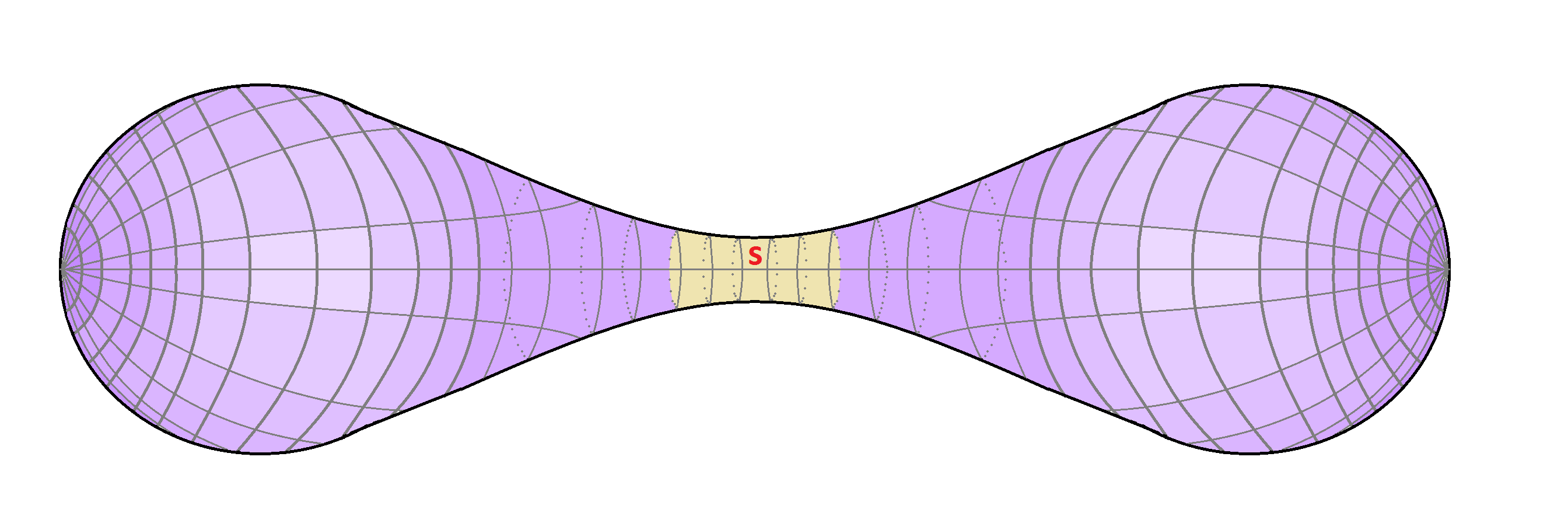} 
		\end{center}
		\caption{\label{fig:bottleneck} The yellow part S represents a small disk on a bottleneck area of some simply connected surface. Even though S is small, it is not simply connected as any loop around the surface lying inside this area is non trivial}.
	\end{figure}
	
	\subsubsection{CSC states: the non geometric case}
	
	Finally, in section \ref{sec:CSCstates}, we derive our main results, Theorems \ref{thm:NonGTrivSStates} and \ref{thm:NonTrivSStates}. Both proofs follows the same lines as in theorems \ref{thm:NonGTrivCubeState} and \ref{thm:nonTrivCubeState} relying on Theorem \ref{thm:getAround} and the assumption of $r$-simple connectedness to ensure that the sets $L(a)$ for $a\in A$ lie inside a ball of small radius around the endpoints $e$ and $f$ as in the cube case.
	
	\subsubsection{Remarks}
	
	It should be noted that the main difference between the asymptotic lower bounds derived from the cube state and from more general CSC states lies in the geometric properties of the underlying topological space. More precisely, the maximal ratio between the radius of a ball and its area (the number of edges lying inside the ball) on the complex, is the main factor influencing the asymptotic behaviour of the circuit depth lower bounds described in this paper. In the euclidean case (e.g the toric code or cube state), this ratio is quadratic in $r$ leading to a stronger lower bound of $\Omega(log(N))$ (and $\Omega(\sqrt N)$ in the geometric case). On the other hand, if this ratio decays exponentially in $r$ (as it is the case for hyperbolic surfaces, see figure \ref{fig:hyperbolic}) we can merely extract a
	$\Omega(log(log(log(N))))$ bound (and $\Omega(log(N))$ in the geometrical case, 
	Theorem \ref{thm:NonGTrivSStates}).

	\subsection{Further discussion}
	Our result is the first circuit depth lower bound for unique groundstates of
	gapped local Hamiltonians, in the challenging setting in
	which entanglement is limited by
	exponential decay \cite{hastingskoma}. 
	
	We note that using Kitaev's circuit-to-Hamiltonian construction one can
	construct a non-trivial state which is very close to the CAT state,
	{\it and} is a unique groundstate
	of a local Hamiltonian; however, this local Hamiltonian is not gapped,
	and thus the state does not exhibit exponentially decaying correlations.
	Indeed, it is exactly those non-decaying corrections which are
	used to prove its non-triviality.
	
	The results
	can be viewed in the more general context of the heirarchy of topological order.
	More specifically, one can consider increasingly growing classes of
	sttaes, which are more and more entangled; at the bottom lie the
	trivial states, then comes the class of unique groundstates of gapped local
	Hamiltonians (UGS), then unique groundstates of local (not necessarily gapped)
	Hamiltonians, and these are contained within the class of states detrermined by their local reduced density matrices among all mixed states (UDA \cite{}),
	and then among
	all pure states (UDP \cite{}). Our proof can be viewed as separating
	the second set from the third; despite previous attempts,
	it remains open 
	to clarify whether the fourth is or is not equal to the fifth.
	This study is thus motivated also by the goal of 
	clarifying the connection between local Hamiltonians,
	local reduced density matrices 
	and the global entanglement determined by them. In the next subsection we provide a thorough description of this heirarchy, for completeness of the context.  
	Finally, we ask: can non-trivial circuit depth lower bound
	be proven for quantum groundstates which are non-homological?
	Of course, a major open question is whether a superlogarithmic
	lower bound can be proven for any quantum state. 
	
	\subsection{The Heirarchy of Topological Order}
	\label{subsec:Hierarchy}
	Our work brings some new insights into 
	the hierarchy of topologically ordered quantum states,
	namely, states with global entanglement.
	In particular, we survey our current understanding of states
	which are {\it not} topologically ordered.
	
	At the bottom of the heirarchy lies the class of trivial states (see definition \ref{def:trivialStates}).
	In \cite{B}, the authors proved that 
	trivial states are unique groundstates of gapped Hamiltonians, and we give an alternative proof in Appendix A.
	We can define the class of states with the aforementioned property: 
	
	\begin{definition} {\bf k,$\Delta-\emph{Gapped-UGS}$ states}
		Let $k>0$. A pure quantum state $\ket{\psi}$ is said to be k,$\Delta-\emph{Gapped-UGS}$ if there exists a k-local Hamiltonian $H=\sum_{i=1}^{m}H^n_i$ with energy gap at most $\Delta$ such that $\ket{\psi}$ is the unique ground state of $H$.
	\end{definition}
	
	Of course, these states are a subclass of unique groundstates of
	Hamiltonians which are not necessarily gapped: 
	
	\begin{definition} {\bf k-UGS states}
		Let $k>0$. A pure quantum state $\ket{\psi}$ is said to be k-\emph{UGS} if there exists a k-local Hamiltonian $H=\sum_{i=1}^{m}H^n_i$  such that $\ket{\psi}$ is the unique ground state of $H$.
	\end{definition}
	
	While k,$\Delta$-Gapped-UGS trivially implies k-UGS, the reverse implication doesn't hold. In fact, 
	in Appendix A, we provide a family of pure states $\{\ket{\phi_n}\}$ which asymptotycaly appoximate the familly $\{\ket{CAT^+}_n\}$ of CAT states. We show that $\{\ket{\phi_n}\}$ is k-UGS but not k,$\Delta$-Gapped-UGS for any constant $\Delta$.
	
	The fact that states are unique groundstates of $k$-local Hamiltonians,
	stems from the question of whether their reduced local density matrices
	can be uniquely ``lifted'' to a single pure state. 
	The following two definitions attempting to capture this notion
	in two slightly different ways 
	were introduced in \cite{ChenDawkins}:
	
	\begin{definition} {\bf k-UDA states}
		Let $k>0$. A pure quantum state $\ket{\psi}$ is said to be \emph{k-UDA} if given any mixed state $\rho$ which satisfies that if for all subsets $A$ of qubits, $|A|=k$
		\[
		Tr_{\bar{A}} (\ketbra{\psi}{\psi}) = Tr_{\bar{A}} (\ketbra{\phi}{\phi})
		\]\\
		then $\ketbra{\psi}{\psi}=\rho$.   
	\end{definition}
	
	\begin{definition} {\bf k-UDP states}
		Let $k>0$. A pure quantum state $\ket{\psi}$ is said to be \emph{k-UDP} if given a pure state $\ket{\phi}$ which satisfies that if for all subsets $A$ of qubits, $|A|=k$
		\[
		Tr_{\bar{A}} (\ketbra{\psi}{\psi}) = Tr_{\bar{A}} (\ketbra{\phi}{\phi})
		\]\\
		then $\ket{\phi}=\ket{\psi}$.   
	\end{definition}
	
	While k-UDA always implies k-UDP (namely, k-UDA is contained in k-UDP),
	the converse is not known to hold: in \cite{XinLu}, the authors showed a $4$ qubit state which is 2-UDP but not 2-UDA; 
	extending this to a family of $n$ qubit states for growing $n$ is left open. 
	
	One can easily show that k-UGS states are also k-UDA (hence also k-UDP)
	since the energy of a state relative to a k-local hamiltonian depends only on its reduced density matrices on sets of k qubits (see Appendix C). 
	The reverse implication is still an open problem (though it was claimed at some point to hold, 
	following from \cite{maxEntro} which was then found wrong in \cite{ConvexChenJi}).

	The above sequence of sets thus
	provides a hierarchy of states which are in some sense increasingly
	``more and more globally entangled'';
	As explained above,
	most of these containments were already known to be strict,
	as can be seen in Figure
	\ref{fig:TopOrder}. 
	Indeed, Theorem \ref{thm:nonTrivCubeState}
	can in fact be put in this context: it shows that the containment of
	trivial states inside Gapped-UGS is also strict, as the spherical code is
	a unique groundstate of the gapped Hamiltonian of the spherical code,
	while it is non-trivial. This containment might be somewhat
	surprising given that Gapped-UGS are known to have exponentially decaying
	correlations, and thus of limited global entanglement.
	This leaves only one open question, regarding the strict containments
	in this heirarchy: Is there a k-UDA state which is not k-UGS?
	In addition, also the question of relation between UDP and UDA requires further clarification. 
	
	We mention that outside of this Heirarchy of {\it lack of} topological order, 
	there is also a separate heirarchy of the more distinguished  topologically ordered states. 
	Such a heirarchy was suggested in \cite{TQO}, which defined the TQO-1 condition:  
	
	\begin{definition} {\bf TQO-1 states (see\cite{TQO})}
		\label{def:TQO}
		A family $\{\ket{\psi_n}\}$ of quantum states on n qubits is said to be TQO-1 if there exists a family of positive rate quantum codes $\{C_n\}$ with macroscopic distance such that for all n, $\ket{\psi_n}\in C_n$.
	\end{definition}
	
	Furthermore, the authors argued that TQO-1 alone could not ensure stability of the related local Hamiltonian under small perturbation (see \cite{TQO} for a formal definition), and defined a more refined class of TO states which satisfy another condition called TQO-2 (see \cite{TQO} for a formal definition); the two conditions together imply such a stable topological order (in Figure \ref{fig:TopOrder} this is denoted 
	STQO). 
	
	%The authors in \cite{TQO} show the existence of a constant treshold $\epsilon>0$ such that given any gapped local Hamiltonian $H_0$ on a lattice satisfying both TQO-1 and TQO-2 and any perturbation $V$ representable as a sum of short-range bounded-norm interactions, the perturbed Hamiltonian $H=H_0+\epsilon V$
	%is still gapped. 

	%As a side remark, observe that TQO-2 doesn't imply TQO-1 and doesn't ensure stability by itself: indeed, the classical 2D ising model in which the gap is unstable under external magnetic field can be shown to satisfy TQO-2. \ynote{STQO is defined as the intersection of TQO1 and TQO2}
	%$\ket{CAT^+}$ and $\ket{CAT^-}$ states on n qubits are identical on any proper subset of $k<n$ qu
	We observe that by definition,  TQO-1 states  are not
	k-UDP. However, the sets do not 
	complement each other; the CAT state is neither TQO-1 nor $k$-UDP. Indeed, the reduced density matrices of the 2 orthogonal CAT states on any proper subset of qubits are similar, but the code they span only has distance 1. 
	
	As in Figure \ref{fig:TopOrder}  we can thus draw the picture as two hierarchies.
	Of course, these two hierarchies can be
	merged into one hierarchy by considering the complementing sets of either of
	them; but this leads to a less intuitive picture we think). 
	In between those hierarchies, lie "intermediate" families of states such as the CAT states which are complicated enough to be outside of k-UDP and can exhibit long range correlations, but are not inside any non-vanishing distance quantum codes.

	The overall hierarchy of topological orders can be visualized
	in Figure \ref{fig:TopOrder}.
	
	\begin{figure}
		\begin{center}
			\includegraphics[scale=0.3]{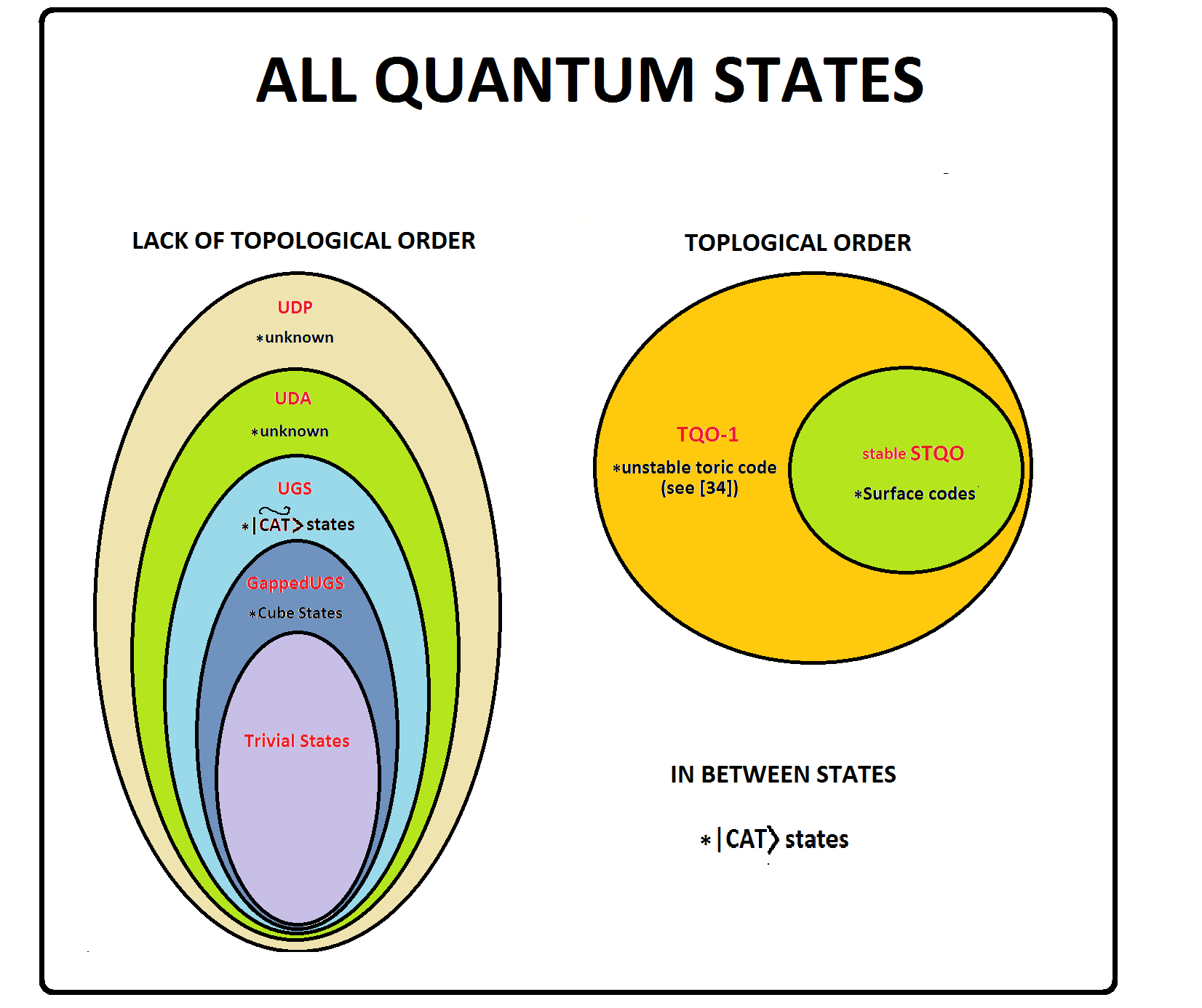} 
		\end{center}
		\caption{\label{fig:TopOrder} The hierarchy of topological orders. On the left side, the hierarchy of lack of quantum topological order. On the right side, the 2 TQO conditions. $\ket{CAT}$ states fall in the middle: They exhibit long range correlations but are not inside any macroscopic distance quantum code.}
	\end{figure}
	
	\section{Background} 
	
	\subsection{Polygonal complexes and graphs}
	
	\begin{definition} {\bf Paths in undirected graphs:}
		\label{def:path}
		Let $G=(V,E)$ be an undirected simple graph and let $e,e'\in E$, we define a path between $e$ and $e'$ as any sequence of
		edges $\gamma=(e=e_{0},e_{1},...,e_{k}=e')$ such that for
		all $0\leq i\leq k-1$, $e_{i}$ and $e_{i+1}$ share a vertex. A path is said to be closed, if $e=e'$. A path $\gamma$ is said to be simple if every vertex belongs to at most 2 edges in $\gamma$.
	\end{definition}

	\begin{definition} {\bf polygonal complex}
		Let $G=(V,E)$ be a finite connected undirected simple graph. Let $F$ be a subset of closed simple paths in $G$ such that for every pair $f_1,f_2\in F$, such that $f_1\neq f_2$, one of the following conditions holds:
		\begin{enumerate}
			\item $f_1\cap f_2=\emptyset$
			\item $f_1\cap f_2\in V$
			\item $f_1\cap f_2\in E$
		\end{enumerate}
		then the triplet $(V,E,F)$ is called a polygonal complex. $V$ is called the \emph{vertex set}, $E$ is called the \emph{edge set} and $F$ is called \emph{face set} of $G$.

		%\begin{enumerate}
		%\item V is a finite discrete set of point in K.
		%\item E is a finite set of embedings $e_i:[0,1]\rightarrow K$, such that $e_i(0),e_i(1)\in V$
		%\item F is a finite set of n-gons  (for varying n) such that for each n-gon $k_n$ in the set, there exists an embeding  $f_i:k_n\rightarrow K$, such that for each side $e$ of $K_n$, the restriction $\widehat{f}_i = f_i|_e$ verifies that $\widehat{f}_i \circ\theta_e \in E$.
		%\end{enumerate}
		%Moreover we require the set  $\{Im(f_i)\}_{f_i\in F}$ to cover K, that is, 
		%\[
		%\cup_{f_i\in F} Im(f_i)=K
		%\]
	\end{definition}
	
	The following definition is dual to the notion of path in polygonal complexes:
	
	\begin{definition} {\bf co-paths in a polygonal complex:}
		\label{def:Copath}
		Let $G=(V,E,F)$ be a  and let $e,e'\in E$, we define a copath between $e$ and $e'$ as any sequence of
		edges $\gamma=(e=e_{0},e_{1},...,e_{k}=f)$ such that for
		all $0\leq i\leq k-1$, $e_{i}$ and $e_{i+1}$ belong to a common face $f\in F$. A copath is said to be closed, if $e=e'$.  A copath is said to be simple if every face in $F$ contains at most 2 edges from $\gamma$
	\end{definition} 
	
	We now define the vertex and face support of a set of edges:
	\begin{definition} {\bf Vertex and face support of a set of edges:}
		\label{def:support}
		Let $G=(V,E,F)$ a polygonal complex and let $X\subseteq E$. Define the vertex and face support of
		$X$ by:
		\[
		V_{X}=\{v\in V|\,\exists e'\in X,\, v\in e'\}
		\]
		and
		\[
		F_{X}=\{f\in F|\,\exists e'\in X,\, e'\in f\}
		\]
	\end{definition}
	
	In the same fashion, we have: 
	
	\begin{definition} {\bf Edge support of  vertex and face subsets:}
		\label{def:support2} given subsets $Y\subseteq V$ or $Y'\subseteq F$
		we can define the edge support of $Y$ and the edge support of  $Y'$ to be: 
		\[
		E_{Y}=\{e\in E|\exists v'\in Y, v'\in e\}
		\]
		\[
		E_{Y'}=\{e\in E|\exists f'\in Y', e\in f'\}
		\]	
		Similarly we define the face and vertex supports of $Y$ and $Y'$ to be:
		\[
		F_{Y}=\{f\in F|\exists e\in E \exists v\in Y, v\in e, e\in f\}
		\]
		\[
		V_{Y'}=\{v\in V|\exists e\in E \exists f'\in Y', v\in e, e\in f'\}
		\]
		
	\end{definition}
	
	We will be interested in polygonal complexes that mimic compact surfaces with no boundary. In the continuous model, this condition has a local characterization: namely, that any points of the 2-dimensional manifold has some neighbourhood homeomorphic to the plane $\mathbb{R}^2$. The next definition gives a discrete analogue to this local condition:
	
	\begin{definition} {\bf closed surface complex}
		\label{def:CSC}
		A closed surface complex (or CSC) is a connected and simply connected polygonal complex $G=(V,E,F)$ such that for all $v\in V$, the following conditions hold:
		\begin{enumerate}
			\item $|E_{\{v\}}|=|F_{\{v\}}|\geq 2$
			\item There are orderings $E_{\{v\}}=(e_1,...,e_k)$ and $F_{\{v\}}=(f_1,...,f_k)$ such that for all i, $f_i\cap f_{i+1} = e_i$ and $f_i\cap f_j = \emptyset$ for all i,j such that $|i-j|\geq 2$ (Index summation is done modulo k).
		\end{enumerate}
		See Figure \ref{fig:clockwise} for a graphic illustration of the CSC condition.
	\end{definition} 
	
	\begin{figure}
		\begin{center}
			\includegraphics[scale=0.3]{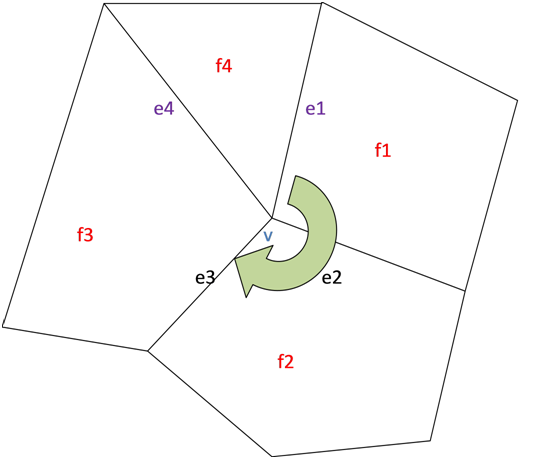} 
		\end{center}
		\caption{\label{fig:clockwise} The CSC condition implies that one can order all edges adjacent to v in a cyclic order such that every 2 consecutive edges share exactly one face}
	\end{figure}
	
	\begin{claim}
		\label{cl:FaceIntersect}
		Let $G=(V,E,F)$ be a CSC. Then for every edge $e\in E$, there are exactly 2 faces $f_1,f_2\in F$ such that $e\in f_1$ and $e\in f_2$
	\end{claim}
	
	\begin{proof}
		Let $v\in e$. Since $G$ is a CSC, we have $E_v=(e_1,...,e_k)$ and there is a face between every 2 cyclically consecutive edges in the list. Since $e\in E$, there exist some index $i\leq k$ such that $e=e_i$ and therefore there exist 2 unique faces $f_1$ and $f_2$ such that $f_1$ connects between $e$ and $e_{i-1}$ and $f_2$ connects between $e$ and $e_{i+1}$. But observe that if there was another face $f_3$ such that $e\in f_3$ then $f_3$ must contain an edge $e'\neq e$ such that  $v\in e'$ and $e'\in E_v$ contradicting the uniqueness part of the CSC condition.
	\end{proof}
	
	\begin{definition} {\bf Vertex and face degree}
		\label{def:VertFaceDeg}
		Let $G=(V,E,F)$ be a polygonal complex. We define the vertex and face degrees of G to be:
		\[
		deg(G)=max_{v\in V} |E_{\{v\}}|
		\]
		and
		\[
		\widehat{deg}(G)=max_{f\in F} |E_{\{f\}}|
		\]
	\end{definition}
	
	Now, consider a polygonal complex $G=(V,E,F)$. We define two natural metrics on $E$:
	
	\begin{definition} {\bf The set of all paths between two edges, $\Gamma$:}
		\label{def:gamma}
		We denote the set of all paths between $e$ and $f$ by $\Gamma_{e,f}$,
		and the set of copaths between $e$ and $f$ as $\widehat{\Gamma}{}_{e,f}$. Note that $|\gamma|=k$ is the number of edges in $\gamma$.
	\end{definition}
	
	The above definitions yield the following 2 discrete metrics on $E$:
	
	\begin{definition}{\bf Metrics:}
		\label{def:metric}
		
		\[
		\mbox{\ensuremath{\forall}}e,f\in E:\, d(e,f)=min_{\gamma\in\Gamma_{e,f}}(|\gamma|-1)
		\]
		\[
		\mbox{\ensuremath{\forall}}e,f\in E:\,\widehat{d}(e,f)=min_{\gamma\in\widehat{\Gamma}{}_{e,f}}(|\gamma|-1)
		\]
	\end{definition}
	
	Note that while both metrics $d$ and $\widehat{d}$ can be defined on any $2$-dimensional polygonal complex, they need not be equivalent (related up to 
	a constant) in general. 
	On the other hand, one can easily show that on CSC's with bounded vertex and face degrees, they are indeed
	equivalent. Indeed, assuming $\widehat{deg}(G) = D$, it
	can be shown that for all $e,f\in E$:
	
	\[
	\widehat{d}(e,f) \leq d(e,f) \leq D \cdot \widehat{d}(e,f) 
	\]
	
	We will not be making use of this claim, but give it just for intuition, hence we skip the proof here. 
	
	\begin{definition} {\bf Topological balls and diameter:}
		\label{def:BallDiam}
		Let $G=(V,E,F)$ be a CSC and let $n\geq 0$. Given some edge $e\in E$ define the ball of radius n around $e$ relative to the metrics d and $\widehat{d}$ by:
		\[
		Ball(e,n)=\{e'\in E|d(e,e')\leq n\}
		\]
		\[
		\widehat{Ball}(e,n)=\{e'\in E|\widehat{d}(e,e')\leq n\}
		\]
		The resulting 2 diameters of $G$ are defined as:
		\[
		diam(G)=max_{e,e'\in E}\, d(e,e')
		\]
		\[
		\widehat{diam}(G)=max_{e,e'\in E}\, \widehat(d)(e,e')
		\]
	\end{definition}
	
	We will also need to define the notions of boundary and coboundary:
	\begin{definition} {\bf Boundary and coboundary:}
		\label{def:BoundCobound}
		Let $G=(V,E,F)$ a CSC and $E'\subseteq E$ be a subset
		of edges. The edge coboundary of $E'$ in $G$ is defined as: 
		\[
		\widehat{\partial}(E')=\{(u,v)\in E\,|\exists(u,v')\in E',\,\nexists(u',v)\in E'\}
		\]
		Observe that from claim \ref{cl:FaceIntersect} every edge $e\in E$ is the intersection of 2 unique
		faces $f_{e}\neq g_{e}\in F$. Therefore, we can define the edge boundary
		of $E'$ as:
		\[
		\partial(E')=\{e\in E\,|\exists e'\in E':e'\in f_{e},\,\nexists e''\in E':e'\in g_{e''}\}
		\]
	\end{definition}
	
	\subsection{Algebraic Topology} 
	We recall some basic definitions from algebraic topology (\cite{Hatcher}).
	
	\begin{definition} {\bf Chain complex}
		A chain complex $C=((C_i),(\partial_i))$ is a sequence of abelian groups $C_i$ and group homomorphisms $\partial_i:C_i\rightarrow C_{i-1}$ called \emph{boundary maps} such that for $i\geq 1$, we have 
		\[
		\partial_{i-1} \circ \partial_i=0 
		\]
	\end{definition}
	
	\begin{definition} {\bf Cochain complex:}
		A Cochain complex $C^*=((C^*_i),(\partial^i))$ is a sequence of abelian groups $C^*_i$ and group homomorphisms $\partial^i:C_i\rightarrow C_{i+1}$ called \emph{coboundary maps} such that for $i\geq 0$, we have 
		\[
		\partial^{i+1} \circ \partial^i=0 
		\]
	\end{definition}
	
	\begin{definition} {\bf Chain and cochain complexes from CSCs:}
		\label{def:ChainComplex}
		Let $G=(V,E,F)$ be a CSC. There exists a natural way to define a both a chain and a cochain complex from G the following way: 
		Define $C_0=<V>$, $C_1=<E>$ and $C_2=<F>$ to be the free abelian groups generated by the sets $V$, $E$, and $F$ respectively. There are 2 well defined group homomorphisms  $\partial_2:C_2\rightarrow C_1$ and $\partial_1:C_1\rightarrow C_0$ corresponding to "taking the boundary": more precisely, we can define $\partial_2$ and $\partial_1$ on basis elements the following way:
		\[
		\forall f\in F: \partial_2(f)=\sum_{e \in f}e
		\]
		\[
		\forall e={v_1, v_2}\in E: \partial_1(e)=v_1+v_2
		\]
		Define $C=((C_i),(\partial_i))_{i\in\{1,2\}}$. It can be easily shown that the chain complex condition $\partial_1 \circ \partial_2 = 0$ holds for $C$. Indeed, if $f=(e_1,...,e_k)$ then $\partial_2(f)=e_1+...+e_k$ and assuming $e_i$ shares the vertex $v_i$ with $e_i+1$ cyclically, then
		\[
		\partial_1\circ \partial_2 (f) = \partial_1 (e_1)+...+\partial_1(e_k) = 2(v_1+...+v_k)=0
		\]
		In a similar fashion, one can naturally define a cochain complex structure on G:\\
		Let $i\in {0,1,2}$, and let $C^*_i=Hom(C_i,Z_2)$ where $Hom(C_i,Z_2)$ is the group of homomomorphisms from 
		$C_i$ to the 2 elements group $Z_2$. Note that for every element $x$ in in a generating set $X_i$ of $C_i$, one can define the following element $f_{x_i}\in C^*_i$:
		\[   
		f_{x_i}(x) = 
		\begin{cases}
		0 & \text{if } x_i=x\\ 
		1 & \text{if } x\neq x\\
		\end{cases}
		\]
		In fact, it can be easily checked that the set $\{f_{x_i}| x_i\in X_i\}$ generates the whole group $C^*_i$. \\
		The boundary maps $\partial_1$ and $\partial_2$ on $C_1$ and $C_2$ induce coboundary maps $\partial^0:C^*_0\rightarrow C^*_1$ and $\partial^1:C^*_1\rightarrow C^*_2$. Let $i\in {1,2}$ and let $f\in C^*_i$ and $x\in C_{i+1}$. Define:
		\[
		[\partial^i(f)](x) = f(\partial^{i+1}(x))
		\]
		
		or in other words, $\partial^i=f\circ \partial^{i+1}$. Now, define $C^*= ((C^*_i),(\partial^i))_{i\in\{1,2\}}$. It is easily seen that the boundary condition $\partial_1 \circ \partial_2 = 0$ induces the coboundary condition on $(\partial^1,\partial^2)$: 
		\[
		\partial^{i+1} \circ \partial^i=0 
		\] 
		giving $C^*$ a structure of cochain complex.
	\end{definition}
	
	In order to extract meaningful invariants from a surface chain complex, one needs to look at the actual homology and cohomology groups derived from $C$ and $C^*$:
	
	\begin{definition} {\bf Homology groups:}
		\label{def:Hom}
		Let $C=(C_i, \partial_i)_{i=0}^n$ be a chain complex.
		The i-th homology group $H_i(C)$ of C is defined as:
		\[
		H_i(C)=Ker(\partial_i) / Im(\partial_{i+1})
		\]\\
		Elements of $Im(\partial_{i+1})$ are called $i$-boundaries and elements of $Ker(\partial_i)$ are called $i$-cycles. 
		When $C$ is derived from a CSC G, we use the notation $H_i(G)=H_i(C)$.
	\end{definition}
	
	\begin{definition} {\bf Cohomology groups:}
		\label{def:CoHom}
		Let $C^*=(C^i, \partial^i)_{i=0}^n$ be a cochain complex.
		The i-th cohomology group $H^i(C^*)$ of $C^*$ is defined as:
		\[
		H^i(C)=Ker(\partial^i) / Im(\partial^{i-1})
		\]\\
		Elements of $Im(\partial^{i-1})$ are called $i$-coboundaries and elements of $Ker(\partial_i)$ are called $i$-cocycles. 
		When $C^*$ is derived from a CSC G, we use the notation $H^i(G)=H^i(C)$.
	\end{definition}
	
	One can interpret the rank of $H_i$ as the number of i-dimensional "holes" in S. A classic example is the torus $T$. The first homology group $H_1(T)$ of $T$ has rank 2. A natural basis for $H_1(T)$ is the one spanned by a lateral and a longitudinal non-trivial loops around $T$. 
	
	\begin{definition} {\bf Homology class:}
		\label{def:HomClass}
		Let $G$ be a CSC and let $\gamma$ be
		a path. Define the \emph{homology class} $[\gamma]_{G}$ of $\gamma$ to be the set of chains $\gamma' \in C_1(G)$ such that $\gamma+\gamma'$ is a 1-boundary (or equivalently,
		$\gamma+\gamma'$ vanishes in $H_{1}(G)$). When the underlying complex
		$G$ is known from  context we just write $[\gamma]$ instead of
		$[\gamma]_{G}$\\
	\end{definition}
	
	In the rest of this paper, we will look at paths (and copaths) as both subsets of the edge set $E$ \emph{and} chains in $C_1$. With that in mind, the expression $\alpha + \beta$  should be understood either as the sum of 2 chains in $C_1$ or equivalently, as the symmetric difference between two sets.
	
	\begin{remark}
		The fact we used the same term $\partial$ for definitions \ref{def:BoundCobound} and \ref{def:Hom} is not coincidental: indeed, it can easily be shown that for any set $E'\subseteq E$ of edges:
		\[
		\partial(E')=\partial_2(F_{E'})
		\] 
		and
		
		\[
		\widehat{\partial}(E')=\partial^0(V_{E'})
		\] 
	\end{remark}
	\subsection{Quantum surface codes}
	\label{Subsec:SurfaceCodes}
	
	First recall the definition of stabilizer codes:
	\begin{definition} {\bf Stabilizer formalism:}
		\label{def:Stab}
		Denote by $P_{n}$ the Pauli group acting on the Hilbert space $(H_n=\mathbb{C}^2)^{\otimes n}$ of n qubits and let $\mathfrak{G}<P_n$ be an abelian subgroup of $P_n$ such that $-1\notin \mathfrak{G}$. Then the set
		\[
		C_{\mathfrak{G}}=\{\ket{\psi}\in H_n|\forall g\in\mathfrak{G}:g\cdot \ket{\psi}=\ket{\psi}\}
		\]
		
		is the $\mathfrak{G}$-stabilized subspace of $H_n$ and has dimension $2^{n-dim(\mathfrak{G})}$ where $dim(\mathfrak{G})$
		is the dimension of $\mathfrak{G}$ as a vector space over $\mathbb{F}_2$.
	\end{definition}
	
	We will also need the following fact about Quantum Error Correcting codes: 
	
	\begin{fact}\label{fact:qecc}{\bf adapted from \cite{nielsen} p.436}: 
		For any two orthogonal states in a quantum error correcting code 
		whose distance is $d_n$, we have 
		for any operator $O_n$ of support on a set $K\subseteq [n]$ of size $|K|<d_n$ 
		
		\[
		\braket{\psi_n|O_n}{\phi_n}=0
		\]
		
		and from the properties of the partial trace, we get:
		
		\[
		Tr_{\bar{K}} (\ketbra{\psi_n}{\psi_n}) = Tr_{\bar{K}} (\ketbra{\phi_n}{\phi_n})
		\]
	\end{fact}
	
	One of the most famous example of stabilizer codes is the toric code, introduced by Kitaev \cite{ToricCode}. 
	Consider the lattice $\mathbb{Z}_n \times \mathbb{Z}_n$ where $\mathbb{Z}_n$ is the cyclic group of order n. For every edge in the lattice define a 2 dimensional site (namely, a qubit). Moreover, for each vertex $v$ in the grid, we associate a \emph{star} Pauli operator $A_{v}=\sigma_{X}^1 \sigma_{X}^2 \sigma_{X}^3 \sigma_{X}^4$ where
	$\{\sigma_{X}^i\}_{i=1}^4$ are applied on the 4 sites (edges) connected to $v$. Similarly, for each face f (any basic $1\times 1$ square), we associate a \emph{plaquette} Pauli operator $B_{f}=\sigma_{Z}^1 \sigma_{Z}^2 \sigma_{Z}^3 \sigma_{Z}^4$ where $\{\sigma_{Z}^1\}_{i=1}^4$ are applied on the 4 sites connected to f.
	Kitaev proved that the subgroup generated by ${A_v, B_f}$ satisfies the above 2 conditions, and has dimension n-2 (see \cite{ToricCode}). Hence, the corresponding stabilized subspace has dimension $2^{n-(n-2)}=4$ and can encode 2 qubits.
	One can naturally extend the definition of surface codes on any CSC in 
	the following way:
	
	\begin{definition} {\bf Stabilizer formalism on CSCs:}
		\label{def:StabSurf}
		Let $G=(V,E,F)$ be a CSC. Associate a qubit with 
		each edge $e\in E$. 
		The underlying Hilbert space $H$ is the tensor product of all $n=|E|$ sites. 
		\\
		For each vertex $v\in V$, define the following Pauli operator:
		\[
		A_{v}=\prod_{e\in E, v\in e}\sigma^e_{X}
		\]
		And for each face $f\in F$, define the following Pauli operator:
		\[
		B_{f}=\prod_{f\in F,e\in f}\sigma^e_{Z}
		\]
		Since each vertex has either 0 or 2 common edges with a given face, the above operators commute. Therefore, $\mathfrak{G} = Span(\{A_{v},B_{f}|v\in V,\, f\in F\})$ is an abelian subgroup of $P_n$, and since $-1\notin \mathfrak{G}$, the stabilized subspace of $\mathfrak{G}$ is a quantum code $C_G$ of dimension $2^{n-dim(\mathfrak{G})}$. 
	\end{definition}
	
	The dimension of $C_G$ is directly related to the first homology class and cohomology classes $H_1(G)$ and $H^1(G)$ of $G$ through the following theorem (see \cite{Bombin} for further details):
	
	\begin{theorem}
		Let $G=(V,E,F)$ be a CSC, then $dim(C_G)=|V|-|E|+|F|+2=2 - \chi(G)$.
	\end{theorem}
	
	\begin{proof}
		From the general theory of the Pauli formalism, $dim(C_G)=n-rank(G)$ where n is the number of physical qubits, and we get $n=|V|$. By definition, we have $G=<A_v,B_f>_{v\in V,f\in F}$. Observe that any non trivial relation between those generators must involve only one type stabilizers: either star X operators or Z plaquette operators. Let  $V'\subseteq V$, and assume that $\prod_{v'\in V'}A_{v'}=I$. We claim that $V'=V$ or $V'=\emptyset$. Indeed, suppose $V'\neq \emptyset$. Then there exists some $v_0\in V'$. Now let $v_1\in V$. From the connectedness of CSCs, there is a path $\gamma={e_0,...,e_k}$ joining $v_0$ to $v_1$. We can also look at $\gamma$ as a sequence of  vertices ${v_0=w_0,...,w_{k+1}=v_1}$. We prove by induction on the path that $w_i\in V'$ for every $i\leq k+1$:\\
		By assumption, $v_0\in V'$. Now assume $v_i\in V'$ for some $i\leq k$. Since $P=\prod_{v'\in V'}A_{v'}=I$, the restriction of $P$ to the qubit on the edge $e_i=(v_i,v_{i+1})$ is identity and in particular, there must be an even amount of plaquettes $A_{v'}$, $v'\in V'$ acting non trivially on $e_i$. Since $A_{v_i}$ already acts no trivially on $e_i$, the plaquette $A_{v_{i+1}}$ must participate in P, or equivalently, $v_{i+1}\in V'$.\\
		The same analysis also holds for star generators, therefore, the only non trivial relations between the original generators are:
		\[
		\prod_{v\in V} A_v
		\]
		\[
		\prod_{f\in F} B_f
		\]
		It follows that $rank(G)=|V|+|F|-2$ and we get 
		\[
		dim(C_G)=n-rank(G)=|E|-|V|-|F|+2 = 2 - \chi(G)
		\]
	\end{proof}
	
	\begin{corollary}
		\label{cor:homologyDimensionCSC}
		Let $G=(V,E,F)$ be a CSC, then $dim(C_G)=dim(H_1(G))=dim(H^1(G))$
	\end{corollary}
	
	\begin{proof}
		It is well known that for any hypergraph G, $\chi(G)=dim(H_0(G))-dim(H_1(G))+dim(H_2(G))$ (see \cite{Hatcher}). Since G is connected, we have both $dim(H_0(G))=1$ and $dim(H_2(G))=1$. Therefore from the previous theorem, we get:
		\[
		dim(C_G)=2-\chi(G)=2-(2-dim(H_1(G)))=dim(H_1(G))
		\]
		Furthermore, from Poincare duality, since CSCs are 2 dimensional complexes, we have $dim(H^1(G))=dim(H_{2-1}(G))=dim(H_1(G))$ finishing the proof.
	\end{proof}
	
	Finally, we will make use of the following notation in the rest of this paper: 
	\begin{definition} {\bf Pauli operators on subsets of edges:}
		\label{def:SubsetPauli}
		Let $G=(V,E,F)$ be a CSC and let $E'\subseteq E$. We define $E'_{X}=\prod_{e\in E'}\sigma_{X}^{e} \otimes I_{E\backslash X}$ 
		and $E'_{Z}=\prod_{e\in E'}\sigma_{Z}^{e}  \otimes I_{E\backslash X}$. In particular, if $\gamma$ is a path in $G$,
		then $\gamma_Z$ is defined as the Pauli operator which applies a $\sigma_{Z}$ operator on each site lying on $\gamma$ and identity everywhere else, and likewise for $\gamma_X$
	\end{definition}
	
	\section{Notations and definitions}
	Before we proceed further, we need to state some definitions that will be relevant for the rest of the proof:
	
	\begin{definition} {\bf Circuit induced directed graph:}
		\label{def:CircGraph}
		Let $U$ be a unitary circuit acting on $H$=$\left(\mathbb{C}^{2}\right)^{\otimes n}$
		we define the following directed graph $G_{U}$ associated with $U:$
		the edges are the wires of $U$ (directed according to the arrow of
		time) and the vertices are the quantum gates together with a set of n input vertices and n output vertices of degree one. An edge $e=(u,v)$ connects
		two vertices $v\, and\, u$ if and only if there exists a wire in $U$
		connecting the gate $u$ to the gate $v$. For the i-th qubit $q$ we
		can associate two edges in $G_{U}$: $in(q)$ and $out(q)$ corresponding
		to the i-th input and i-th output wire.
	\end{definition}
	
	With that definition in hand, we can now formally define the upper and lower light cones of any subset of edges from $E$:
	
	\begin{definition} {\bf Upper and lower light cone:}
		\label{def:LC}
		Let $U$ be a unitary circuit acting on and Hilbert
		space $H$=$\left(\mathbb{C}^{2}\right)^{\otimes n}$ of the set $S$ 
		of $n$ qubits. Let $S'\subseteq S$
		be a subset of the qubits. We define the upper and lower light cones of
		$S'$ to be: 
		\[
		L^{\uparrow}(S')=\{q\in S|\exists q'\in S',\,\Gamma_{in(s'),out(s)}\neq\emptyset\}
		\]
		\[
		L^{\downarrow}(S')=\{q\in S|\exists q'\in S',\,\Gamma_{in(q),out(q')}\neq\emptyset\}
		\]
	\end{definition}
	
	The following easy fact follows: 
	
	\begin{fact}\label{fact:equivalence}
		For two sets of qubits, $W$ and $V$, we have 
		\[ L^\uparrow(W)\subseteq V ~~~ iff ~~~~W\subseteq L^\downarrow(V)\]
		\[ V\subseteq L^\uparrow(W) ~~~ iff ~~~~L^\downarrow (V)\subseteq W\]
	\end{fact} 
	\begin{proof} 
		Observe that $L^\uparrow(W)\subseteq V$ iff for all 
		input qubit $w\in W$, there is a path in the graph from $w$ to an output $v\in V$ 
		and this is iff $W\subseteq L^\downarrow(V)$. The other direction is similar. 
	\end{proof} 
	
	We are now ready to define the \emph{effective support} of a given path 
	$\gamma$ in the CSC; 
	We treat $\gamma$ as a subset of edges, but also as a subset of the qubits. 
	
	\begin{definition} {\bf The effective supports: A and B:}
		\label{def:EffSupp}
		Let $G=(V,E,F)$ be a CSC, and $U$ be a unitary
		circuit acting on $E$ such that $U|0\rangle^{\otimes|E|}=|\psi\rangle$
		where $|\psi\rangle\in C_{G}$ (Here, $C_G$is  the quantum code space defined on G as in definition \ref{def:StabSurf}). Let $e,f\in E$. For each homology
		class $[\gamma]$  in $\Gamma_{e,f}$, define the \emph{lower effective
			support} of $[\gamma]$ under the action of $U$ to be: 
		\[
		A_{[\gamma],U}=\bigcap_{\gamma'\in[\gamma]}L^{\downarrow}(\gamma')
		\]
		We also define the \emph{upper effective support} $B_{[\gamma],U}$ to be the upper light cone of $A$ under the action of $U$:
		\[
		B_{[\gamma],U}=L^{\uparrow}\left(A_{[\gamma],U}\right)
		\]
	\end{definition}
	
	Observe that by definition, $A_{[\gamma],U}$ is a function of the whole homology class $[\gamma]$ and doesn't depend on any particular choice of representative chain  $\gamma'\in[\gamma]$
	For clarity, when both $\gamma\in\Gamma_{e,f}$ and $U$ are fixed, we will write $A=A_{[\gamma],U}$ and $B=B_{[\gamma],U}$.
	We can now proceed to prove a weak version of our main result.

	\section{The starting point: Bravyi's commuting operators lemma}
	\label{sec:Bravyi}
	
	We now prove a lemma that we will be using multiple times
	in the rest of this paper. The lemma is the main idea underlying Bravyi's unpublished proof
	of a $\Omega(\sqrt{n})$ lowerbound for the circuit depth of toric code states,
	in the geometric case \cite{B}. 
	
	We stress that the lemma does {\it not} rely on any assumption on the geometry of the quantum circuit.

	\begin{lemma}{\bf Operators out of effective support of $\gamma$ commute 
			with operators on $\gamma$}
		\label{lem:ComutOp}
		Let $G$ be a CSC, and let $U$ be a quantum circuit such that $U\ket{0}^n=\ket{\psi}\in C_G$. Let $e,f\in E$, let $\gamma\in\Gamma_{e,f}$ and let $B=B_{[\gamma],U}$. Then, for every operator $P$ supported on $E\backslash B$ that stabilizes $|\psi\rangle$,
		we have:
		\[
		P\gamma{}_{Z}|\psi\rangle=\gamma{}_{Z}P|\psi\rangle. 
		\]
	\end{lemma}
	
	\begin{proof}
		We first prove that in the $U$ basis, 
		applying $\gamma$ on the groundstate $\ket{\psi}$ 
		can be replaced by applying an operator whose support is confined to $A=A_{[\gamma],U}$. 
		In other words, 
		$U^{\dagger}\gamma{}_{Z}U|0\rangle^{n}=
		|0\rangle_{E\backslash A}\otimes|\phi\rangle_{A}$ for some state $\ket{\phi}$. 
		To show this, let $e$ be an edge such that $e\notin A$. 
		Then by definition, there exists a path $\alpha\in[\gamma]$ such that
		$e\notin L^{\downarrow}(\alpha)$. Since $\gamma$ and $\alpha$ are
		both in $[\gamma]$, we have $\gamma{}_{Z}\alpha_{Z}\in\mathfrak{G}$
		and therefore:
		\[
		\gamma{}_{Z}\alpha_{Z}|\psi\rangle=|\psi\rangle
		\]
		so
		\[
		\alpha_{Z}|\psi\rangle=\left(\gamma{}_{Z}\right)^{-1}|\psi\rangle=\gamma{}_{Z}|\psi\rangle
		\]
		\[
		U^{\dagger}\alpha_{Z}U|0\rangle^{n}=U^{\dagger}\gamma{}_{Z}U|0\rangle^{n}
		\]
		Since $U^{\dagger}\alpha_{Z}U$ has support on $L^{\downarrow}(\alpha)$,
		and $e\notin L^{\downarrow}(\alpha)$, it leaves $|0\rangle_{e}$ intact.
		Therefore:
		\begin{eqnarray}
		U^{\dagger}\gamma{}_{Z}U|0\rangle^{n} & = & U^{\dagger}\alpha_{Z}U|0\rangle^{n}\\
		& = & U^{\dagger}\alpha_{Z}U(|0\rangle_{e}\otimes|0\rangle_{E\backslash\{e\}})\\
		& = & |0\rangle_{e}\otimes\hat{\alpha_{Z}}|0\rangle_{E\backslash\{e\}}
		\end{eqnarray}
		Since this is true for all $e\notin A$, we have 
		\[
		U^{\dagger}\gamma{}_{Z}U|0\rangle^{n}=|0\rangle_{E\backslash A}\otimes|\phi\rangle_{A}
		\]
		Now since $P$ is supported on $E\backslash B$, $U^{\dagger}PU$
		is supported on $L^{\downarrow}(E\backslash B)=L^{\downarrow}(E\backslash L^{\uparrow}(A))\subseteq E\backslash A$ (where the last inequality follows by Fact  \ref{fact:equivalence}  from 
		$E\backslash L^{\uparrow}(A)\subseteq E\backslash A \subseteq L^\uparrow(E\backslash A)$.
		So $U^\dagger PU$'s support is contained in $E\backslash A$. 
		
		Since $P|\psi\rangle=|\psi\rangle$ we have $U^{\dagger}PU|0\rangle^{n}=|0\rangle^{n}$.
		Since $U^{\dagger}PU$ has support on $E\backslash A$, it also holds
		that for every pure state $|\phi\rangle_{A}$ of the qubits in $A$, 
		\[
		U^{\dagger}PU(|0\rangle_{E\backslash A}\otimes|\phi\rangle_{A})=|0\rangle_{E\backslash A}\otimes|\phi\rangle_{A}
		\]
		Therefore, 
		\begin{eqnarray}
		U^{\dagger}P\gamma_{Z}|\psi\rangle & = & (U^{\dagger}PU)(U^{\dagger}\gamma{}_{Z}U)|0\rangle^{n}\\
		& = & (U^{\dagger}PU)(|0\rangle_{E\backslash A}\otimes|\phi\rangle_{A})\\
		& = & |0\rangle_{E\backslash A}\otimes|\phi\rangle_{A}\\
		& = & (U^{\dagger}\gamma{}_{Z}U)|0\rangle^{n}\\
		& = & U^{\dagger}\gamma{}_{Z}|\psi\rangle
		\end{eqnarray}
		and it follows that
		\[
		P\gamma{}_{Z}|\psi\rangle=\gamma{}_{Z}|\psi\rangle =  \gamma{}_{Z}P|\psi\rangle.   
		\]
		
	\end{proof}

	\section{Geometric non-triviality of the cube states}
	\label{sec:geocube}
	
	In order to understand the ideas behind the proof of the general case, we will first prove a weaker version of the main theorem on a simple CSC: the cube. There are two motivations for this choice: 
	On the one hand, there exists a fairly simple and natural way to define a CSC family on the cube as we shall soon see.
	On the other hand, the cube has a trivial first homology group. Hence, the stabilized subspace arising from the surface code construct mentioned earlier has dimension 1. Therefore, that unique state doesn't satisfy the usual Topological Order condition (see definition \ref{def:TQO}). Furthermore, there are no long range correlations in the cube state as every logical operator is in fact trivial and measurement of the logical qubits doesn't yield any information on the state nor modifies it.
	Hence, we have to use additional tools in order to show non-triviality of the cube states. 
	
	\begin{definition} {\bf The cube states:}
		\label{def:CubeState}
		Let $S_n=(V_n, E_n, F_n)$ be the CSC such that 
		\[
		V_n = [n]\times [n]
		\]
		\[
		E_n = \{((i,j),(i',j'))\,|\,(i=i'\, and\, |j-j'|=1)\, or \, (j=j' \, and \, |i-i'|=1)\}
		\]
		and, $F_n$ is defined to be the set of all $1$ by $1$
		squares with edges from $E_n$. \\
		Now, take 6 copies of $S_n$ and "glue" them together to get a cube, when we identify edges and vertices from different copies of $S_n$ if and only if they coincide after the gluing process.
		
		We call the resulting CSC: $T_n$. For each edge in $T_n$ we associate a qubit, and we  define $C_n$  the surface code associated to $T_n$ according to the
		previous quantum code construction (see definition \ref{def:StabSurf}).
		Let $N=O(n^2)$ to be the total number of qubits in $C_n$. 
		Since the first homology group of the cube vanishes, it follows from \cite{ToricCode} 
		that $C_n$ contains a single state $\ket{\psi_n}$. We define $\ket{\psi_n}$ 
		as the $n$-th cube state defined on $N$ qubits.
	\end{definition}
	
	In order to prove Theorem \ref{thm:NonGTrivCubeState} 
	we consider the cube state $\ket{\psi_n}$ defined on the CSC $T_n$ as defined above. Given 2 edges e and f, we first show that the effective support of any path $\gamma\in \Gamma{e,f}$ is bounded inside a small region close to either e or f.
	
	\begin{lemma}{\bf Shallow circuit implies small effective support for 
			operators on $\gamma$}
		\label{lem:SausageIntersection}
		Let $n\geq 0$ and let $\ket{\psi_n}$ the $n$-th 
		cube state as defined above. Let $e,f\in E_n$ and $\gamma\in \Gamma_{e,f}$. 
		Furthermore, assume there exists a quantum circuit $U_n$ of depth $d_n$ using geometrically local gates of locality $c$ (namely, two qubits acted upon 
		by the same gate are within distance $c$), 
		such that $U_n\ket{0^N}=\ket{\psi_n}$. Let $A=A_{[\gamma],U_n}$ and $B=B_{[\gamma],U_n}$ be the lower and upper effective supports on $T_n$ associated with $U_n$ and $\gamma$. Then for large enough $n$, and assuming $d=o(n)$,
		we have $A\subseteq \widehat{Ball}(e, cd_n)\cup \widehat{Ball}(f, cd_n)$ 
		and $B\subseteq \widehat{Ball}(e, 2cd_n)\cup \widehat{Ball}(f, 2cd_n)$. 
	\end{lemma}
	
	\begin{proof}
		Let $\gamma'\in[\gamma]$. Since $U$ only makes use of geometrically local gates, the lower light cone $L^{\downarrow}(\gamma')$ of $\gamma'$ satisfies: 
		
		\begin{equation}\label{eq:conebound}
		L^{\downarrow}(\gamma')\subseteq \{e'\in E|d(e',\gamma')\leq c\cdot d\}
		\end{equation}

		%Indeed, starting from the support $S$ of $\gamma'$ and applying only the first layer from the reversed circuit $U^\dagger$, the lower light cone of $\gamma'$ can only extend its reach to qubits within a distance $c-1$ from $S$. Indeed, since the quantum gates composing $U$ are spatially local, they can only affect qubits within a distance $c-1$ from each other. Applying this bound recursively on D layers proves (3).
		Now let $e'\in E$, and assume that $d(e,e')>cd$ and $d(f,e')>cd$. 
		The ball $\widehat{Ball}(e',cd)$ doesn't contain neither $e$ or $f$. 
		But from the geometry of the cube
		removing a ball of radius $d=o(n)$ from the original cube leaves it connected. 
		Therefore, there exists some path $\gamma''\in [\gamma]$ such
		that $d(\gamma'',e')>cd$ and then from Equation (\ref{eq:conebound}), 
		$e'\notin L^{\downarrow}(\gamma'')$. But obviously, since the cube is simply connected, $\gamma''\in [\gamma]$. From the definition of $A$, we conclude that $e'\notin A$. 
		This proves the first part of the claim: 
		$A\subseteq \widehat{Ball}(e, cd_n)\cup \widehat{Ball}(f, cd_n)$. 
		Since $B= L^\uparrow(A)$, by definition any qubit in $B$ is within distance 
		at most $cd_n$ from any qubit in $A$ and thus $B$ is contained in the set of 
		all qubits of deistance at most $cd_n$ from 
		$\widehat{Ball}(e, cd_n)\cup \widehat{Ball}(f, cd_n)$; this set is contained in 
		$\widehat{Ball}(e, 2cd_n)\cup \widehat{Ball}(f, 2cd_n)$
	\end{proof}
	
	Note that from the above lemma, and assuming $d_n=o(n)=o(\sqrt{N})$, 
	$B$ has to be a  $o(\sqrt{N})$-size set.
	
	We can now provide the proof of Theorem \ref{thm:NonGTrivCubeState}:
	
	\begin{proof} [proof of Theorem \ref{thm:NonGTrivCubeState}]
		Choose two edges $e,f$ lying on two
		opposite faces of the cube $T_n$, and choose some path $\gamma\in\Gamma_{e,f}$. Obviously, since $e$ and $f$ lie on two opposites faces we have $d(e,f)\geq n$. 
		Now assume by contradiction that there exists a depth $d_n=o(\sqrt{N})$ geometrically 
		local quantum circuit $U_n$ such that $U_n\ket{0}^N = \ket{\psi_n}$. W.l.o.g, we 
		assume $U_n$ makes use of quantum gates on at most $c$ qubits. Since $e$ and $f$ lie 
		on opposite faces $F_e$ and $F_f$ of the cube $T_n$, consider the remianing four
		faces cyclically glued to each other. We can now consider the closed co-path 
		cutting all of those four faces in half and in the middle, and call it $\zeta$. 
		Since $\zeta$ is a closed co-path it is a coboundary. 
		Note that $\zeta\cap B = \emptyset$: indeed, by construction, for every edge $e'$ lying in $F_e\cup F_f$, $d(\zeta, e')\geq \nicefrac{n}{2}$. But from lemma \ref{lem:SausageIntersection}, $B\subseteq Ball(e,cd_n)\cup Ball(f,cd_n)$, and since $d_n=o(n)$, we get $d(B,\zeta)>0$. The operator $\zeta_X$ is thus supported on $E_n\backslash B$. Since $\zeta_X$ is an $X$ 
		operator on a coboundary, it follows
		that $\zeta_X$ stabilizes $\ket{\psi_n}$. Therefore, we can apply 
		Lemma \ref{lem:ComutOp} using $P=\zeta_X$ to get: 
		
		\begin{equation}\label{eq:what?}
		\zeta_X \gamma{}_{Z}|\psi_n\rangle=\gamma{}_{Z}|\psi_n\rangle
		\end{equation}
		
		But observe that $\zeta$ actually separates $T_n$ into two
		disconnected components, one containing $e$ and the other one containing f! Therefore, any path $\gamma'$ connecting $e$ to $f$ 
		has to cut through $\zeta$ on an odd number of edges. Hence, $\zeta_X \gamma_{Z} = -\gamma_{Z} \zeta_X$. Inserting this equality in equation \ref{eq:what?}
		we get:
		\[
		-\gamma_{Z}|\psi_n\rangle =-\gamma_{Z} \zeta_X	|\psi_n\rangle=\zeta_X \gamma{}_{Z}|\psi_n\rangle=\gamma{}_{Z}|\psi_n\rangle
		\]
		and we get $|\psi_n\rangle=0$ which obviously doesn't hold. This concludes the proof of Theorem \ref{thm:NonGTrivCubeState}.
	\end{proof}
	
	\section{Non triviality of the cube states}  
	\label{sec:NonTrivCubeState}
	In this subsection we shall generalize our result from the previous chapter and prove Theorem \ref{thm:nonTrivCubeState}. 
	Again, we consider the familly of cube states $\{\ket{\psi_n}\}$ defined on the cubes $\{T_n\}$, but this time we will consider quantum circuits based on gates 
	whose only restriction is to have bounded support size. 
	In particular, we don't require from the quantum gates in a generating 
	circuits $\{U_n\}$ to have support on qubits that are close to each other 
	in the CSC metric. 
	
	We start by explaining in more detail why the previous proof doesn't work 
	in this more general non-geometrically-local setting.  
	The main issue with the previous proof lies in 
	Lemma \ref{lem:SausageIntersection}: without the geometrically local condition on quantum gates,
	we can't bound $|A|$ nor $|B|$ within two small balls 
	around $e$ and $f$. Unfortunately, in the proof of Theorem \ref{thm:NonGTrivCubeState}, 
	we strongly used the fact that 
	$B$ is in some sense "small", as it can be confined in the union 
	of two small radius balls around $e$ and $f$, 
	as in Lemma \ref{lem:SausageIntersection}. 
	This was used for the construction of the path $\zeta$ outside of $B$, 
	required to apply Lemma \ref{lem:ComutOp}. 
	In the non-geometrical case, a similar argument to that of Lemma \ref{lem:SausageIntersection} only implies
	that $|B|=O(diam(T_n))=O(n)=O(\sqrt{N})$; this bound is not 
	strong enough to guarantee a path 
	$\gamma$ separating $e$ from $f$ which does not intersect $B$,
	needed for the application of Lemma \ref{lem:ComutOp}.
	In fact, we later show that we {\it can} assume that $A$ is contained in the union 
	of two small balls around $e$ and $f$ - this is done later in Claim \ref{cl:supportinballs} and
	requires much more work than the analogous Lemma
	\ref{lem:SausageIntersection} in the geometrical case.
	However, the method of generating a path around one of these balls 
	would not work in the non-geometrical case, since the balls no longer 
	have a nice geometrical location. We will need to derive a constradiction 
	via a different argument.

	\subsection{A different approach towards a contradiction: A lower bound on $|A|$}
	To derive a constradiction, we prove the following lemma, 
	Lemma \ref{lem:LargeB}, stating
	that regardless of the depth of the circuit, and without relying on any geometrical restrictions 
	on the gates, the effective support 
	$B$, as well as $A$, must in fact be {\it large}. This is what will lead 
	to a constradiction with the above mentioned Claim \ref{cl:supportinballs} stating
	that $A$ is small. 
	
	\begin{lemma}
		\label{lem:LargeB}
		Let $G=(V,E,F)$ be a CSC such that $|E|=N$ and assume $U\ket{0}^N\in C_G$ where $U$ is a quantum circuit. Let $e,f$ be 2 edges in $E$, $\gamma\in\Gamma_{e,f}$ and $B=B_{[\gamma],U}$. Then there exists $\gamma'\in [\gamma]$ which
		is contained in $B$; in particular, $|B|\geq d(e,f)$.
	\end{lemma}
	
	Note that since this lemma holds 
	also in the special case of geometrically restricted gates, we derive 
	an alternative proof to Theorem \ref{thm:NonGTrivCubeState},
	using the fact that Lemma \ref{lem:LargeB} put together with Lemma \ref{lem:SausageIntersection} leads to  
	a contradiction.

	\begin{proof}
		Suppose by contradiction that $B$ does not contain any $\gamma\in [\gamma]$. Note that
		$e,f\in B$ since $e,f\in\gamma'$ for all $\gamma'\in[\gamma]$.
		By the assumption, $e$ and $f$ belong to a different path-connected
		component of $B$. Let $B_{e}\subseteq B$ be the connected component
		of $B$ such that $e\in B_{e}$. Define $\mu= \widehat{\partial}(B_{e})$.
		We first prove that $\mu_{X}=\prod_{v\in V_{e}}A_{v}\in\mathfrak{G}$.
		Indeed, by definition, if $e'$ is in $\widehat{\partial}(B_{e})$, then $e'$
		is connected to exactly one vertex $v$ in $V_{B_{e}}$ and therefore,
		$(\prod_{v\in V_{B_{e}}}A_{v})|_{e'}=\sigma_X$. On the other hand, if
		$e'\in E\backslash\widehat{\partial}(B_{e})$ then $e'$ is connected to either
		$0$ or $2$ vertices from $V_{B_{e}}$ and 
		therefore, $(\prod_{v\in V_{B_{e}}}A_{v})|_{e'}=I$. 
		Therefore, $\mu_{X}=\prod_{v\in V_{B_{e}}}A_{v}$,
		and it is obviously in $\mathfrak{G}$. \\
		Now we prove that $\mu\cap B=\emptyset$: Indeed, by construction,
		$\mu\cap B_{e}=\emptyset$. Assume by contradiction that there exists
		$e'\in B\backslash B_{e}\cap\mu$. Since $e'\in\mu$, $e'$ has a vertex
		in $V_{B_{e}}$. But then, it follows that $e'\in B_{e}$ and we get a contradiction.\\
		Since $\mu_{X}\in\mathfrak{G}$ we have $\mu_{X}|\psi\rangle=|\psi\rangle$,
		and moreover, $\mu_{X}$ is contained in $E\backslash B$ since $\mu_{X}\cap B=\emptyset.$
		Therefore we can apply lemma \ref{lem:ComutOp}
		with $P=\mu_{X}$, and for all
		$\gamma\in\Gamma_{e,f}$ we get: 
		\[
		\mu_{X}\gamma{}_{Z}|\psi\rangle=\gamma{}_{Z}|\psi\rangle
		\]
		But note that since $e,f$ are in different connected components of $B$,
		then $|\mu\cap\gamma|$ is odd. Indeed, let $\gamma=(e_{1},...,e_{k})$.
		Then we can extract from $\gamma$ a sequence of vertices: $(v_{1},...,v_{k+1})$
		where $e=(v_{1},v_{2})$, $f=(v_{k},v_{k+1})$ and $e_{i}=(v_{i},v_{i+1})$
		for $1\leq i\leq k$. But note that $e_{i}\in\mu$ if and only if exactly one
		of the vertices $v_{i}$,$v_{i+1}$ is in $V_{B_e}$. Therefore, $|\mu\cap\gamma|$
		counts the number of times we get in/out of $V_{B_e}$ in the sequence
		$(v_{1},...,v_{k+1})$. since $v_{1}\in V_{B_e}$ and $v_{k+1}\notin V_{B_e}$,
		$|\mu\cap\gamma|$ has to be odd. \\
		It follows that 
		\[
		\mu_{X}\gamma{}_{Z}=-\gamma_{Z}\mu_{X}
		\]
		so,
		\[
		\mu_{X}\gamma{}_{Z}|\psi\rangle=-\gamma_{Z}\mu_{X}|\psi\rangle=-\gamma_{Z}|\psi\rangle
		\]
		and we get a contradiction.
	\end{proof}
	
	A simple corollary is that provided the circuit depth is small, then $A$ is  also large: 
	\begin{corollary}\label{cor:Asize} 
		Assume the depth of the generating circuit $U_n$ of $\ket{\psi_n}$ is $d_n$, then $|A|> \frac{d(e,f)}{c^d_n}$
	\end{corollary} 
	
	\begin{proof} 
		Since $L^{\uparrow}(A)=B,$ we have $c^{d_n}|A|\geq|B|$,
		and from lemma \ref{lem:LargeB}, 
		$|B|\geq d(e,f)$. Therefore:
		\[
		|A|\geq\frac{d(e,f)}{c^{d_n}}. 
		\]
	\end{proof}
	
	To derive a contradiction in the non-geometrical case, we need the non-geometrical analogue of
	Lemma \ref{lem:SausageIntersection}, providing an upper bound on $|A|$.
	This requires developing some tools, which we do in the next subsection.

	\subsection{Upper bound on $|A|$ using $\gamma$-separation}\label{subsec:upperlightcone}
	
	Here we prove the following claim, which essentially replaces 
	Lemma \ref{lem:SausageIntersection} in the geometrical case, 
	stating that the effective support $A$ is contained in two small balls 
	surrounding $e$ and $f$, except here the size of the balls is exponentially bigger, due to the
	lack of geometrical restriction, but it is still bounded by a constant. The proof is significantly more
	complex.
	
	\begin{claim}\label{cl:supportinballs}
		Assume that $e$ and $f$ are two edges on opposite sides  in $T_n$ 
		so that $d(e,f)\geq n$. Then
		$A\subseteq L^{\downarrow}(\widehat{Ball}(e,c^{d_n})\cup 
		\widehat{Ball}(f,c^{d_n}))$.
	\end{claim}
	
	To prove this claim, we need to define a new notion: 
	\begin{definition} {\bf $\gamma$-separation:}
		\label{def:GammaSep}
		Let $G=(V,E,F)$ be a CSC and let $X\subseteq E$. Let $e,f\in E$ and $\gamma\in\Gamma_{e,f}$. $X$ is called $\gamma-separating$
		if for all $\gamma'\in[\gamma]$, $\gamma'\cap X\neq\emptyset$.
	\end{definition}
	The main motivation behind this definition is encompassed in the following two claims:
	
	\begin{claim}{\bf The upper light cone of any edge in the lower effective support 
			A is $\gamma$-separating}
		\label{cl:LCaGammaSep}
		Let $G=(V,E,F)$ be a CSC  and let $e,f\in E$ and $\gamma\in\Gamma_{e,f}$. Assume that $U\ket{0}^N\in C_G$. If $g\in A=A_{[\gamma], U}$, then $L^{\uparrow}(g)$ is $\gamma-separating$.
	\end{claim}
	
	\begin{proof}
		Assume there exists a path $\gamma'\in[\gamma]$ such
		that $\gamma'\cap L^{\uparrow}(g)=\emptyset$. Therefore, we have
		$g\notin L^{\downarrow}(\gamma')$ but since $g\in A=\bigcap_{\gamma'\in[\gamma]}L^{\downarrow}(\gamma')$,
		it follows that $g\in L^{\downarrow}(\gamma')$ and we get a contradiction. 
	\end{proof} 
	
	\begin{claim}
		\label{cl:GammaSepCon}{\bf $\gamma-separating$ sets can always be reduced to be copath connected}:
		Let $G=(V,E,F)$ be a CSC and let $e,f\in E$ and $\gamma\in\Gamma_{e,f}$. Let $X,X'\subseteq E$ such that $X$ and $X'$ are
		not copath connected to each other (within $X \cup X'$). 
		Furthermore, assume that $X\cup X'$ is $\gamma-separating$ and let $\gamma'\in [\gamma]$. Then either $X$ or $X'$ intersects \emph{both} 
		$\gamma$ and $\gamma'$ non trivially.
	\end{claim}
	
	\begin{proof}
		Since $X\cup X'$ is $\gamma-separating$, we can assume w.l.o.g that $X\cap \gamma\neq \emptyset$ and $X'\cap \gamma'\neq \emptyset$.
		Since $\gamma'\in[\gamma]$, $\gamma'+\gamma$ is a boundary,
		and therefore there is a subset $F'\subseteq F$ such that $\gamma'+\gamma=\sum_{f'\in F'}\partial_F(f')$.
		Define $J=F'\cap F_{X}$ , $\alpha=\sum_{f\in J}\partial_F(f)$ and
		$\beta=\gamma+\alpha$. By definition,
		$\alpha$ is a boundary (as a sum of boundaries) and therefore, $\beta\in[\gamma]$.\\
		
		Assume towards a contradiction that $X\cap \gamma' = \emptyset$ and $X'\cap \gamma = \emptyset$. We will
		show that $\beta\cap(X\cup X')=\emptyset$ contradicting the assumption that $X\cup X'$ is $\gamma-separating$:\\
		Observe that $i\in\beta$ iff $i\in\gamma\backslash\alpha$ or $i\in\alpha\backslash\gamma$.
		Also observe that given an edge $i\in E$, $i$ belongs to 1 face
		exactly in $F'$ iff $i\in\gamma + \gamma'$ . We have 2 cases:
		\begin{enumerate}
			\item First assume that $i\in\gamma\backslash\alpha$. Since $X'\cap\gamma=\emptyset$
			, we have $i\notin X'$. Now if $i\in \gamma'$ then obviously, $i\notin X$, so we can assume that $i\in \gamma\backslash\gamma'$. From the above observation, since $i\in\gamma+\gamma'$,
			i belongs to a unique face $f_i\in F'$. But, if $i\in X$ 
			then we also have $f_i\in F_X$, and we get $f_i\in F'\cap F_X=J$. From the uniqueness of $f_i$, we get $i\in \sum_{f\in J}\partial_F(f) = \alpha$ contradicting the assumption. Therefore we conclude that $i\notin X$, and overall, $i\notin X\cup X'$
			\item Now assume that $i\in\alpha\backslash\gamma$. 
			Since $i\in \alpha$, there exists a face $f_i\in J\subseteq F_X$ such that $i\in f_i$. Since $X$ and $X'$ are not co-path connected to each other, it follows immediately that $i\notin X'$ (otherwise, $X$ and $X'$ would be connected through $f_i$). On the other hand, assume $i\in X$. Since $i\in \alpha$, $i$ belongs to a unique face $f_i\in J$ and since $i\in X$, $i$ belongs to exactly 2 faces $f_i,f_i'\in F_X$. But $f_i'\in F_X\backslash J\implies f_i'\notin F'$. Therefore, $i$ belongs to a unique face $f_i$ in $F'$ and from the above observation,  we get $i\in \gamma + \gamma'$. But by assumption, $i\notin \gamma$ and since $X\cap \gamma'=\emptyset$, and $i\in X$, also $i\notin \gamma'$. Therefore, $i\notin \gamma + \gamma'$ and we get a contradiction to the assumption that $i\in X$. Overall, we proved that $i\notin X\cup X'$.
		\end{enumerate}
		To conclude, we showed that $\beta\cap(X\cup X')=\emptyset$ and therefore
		$X\cup X'$ is not $\gamma-separating$. It follows that either $X\cap \gamma' \neq \emptyset$ or $X'\cap \gamma \neq \emptyset$ which proves the claim.
	\end{proof}
	
	\begin{corollary}
		\label{cor:AisConnected}
		Let $G=(V,E,F)$ be a CSC and let $e,f\in E$ and $\gamma\in\Gamma_{e,f}$. If $X\subseteq E$ is $\gamma-separating$ then
		there exists a copath-connected component $X'\subseteq X$ that is
		also $\gamma-separating$.
	\end{corollary}
	
	\begin{proof}
		Let $X=\bigcup_{i=1}^{m}X_{i}$ where $X_{i}$ are the
		co-path connected components of $X$ . We proceed by induction on the number of co-path connected components of $X$: \\
		If $m=1$, the statement is
		trivial.\\
		Now, assume $m\geq2$: if $X_{1}$ is $\gamma-separating$ then we are done. Otherwise there
		exists some $\gamma'\in[\gamma]$ such that $X_{1}\cap\gamma'=\emptyset$. Let $\gamma''\in [\gamma]$. 
		Since $Y = \bigcup_{i=2}^{j}X_{i}$ and $X_{1}$ are not co-path connected
		and $X=Y\cup X_1$ is $\gamma'-separating$, we can infer from claim \ref{cl:GammaSepCon} that either $X_{1}$ 
		or $Y$ intersects both $\gamma'$ and $\gamma''$ non trivially. From the assumption, $X_{1}\cap\gamma'=\emptyset$ and therefore, $Y\cap \gamma''\neq \emptyset$. Therefore, since this holds for all $\gamma''\in [\gamma]$, we deduce that 
		$Y$ is $\gamma-separating$. 
	\end{proof}
	
	From the above corollary, one can always extract a \emph{connected}, $\gamma-separating$ set from a non connected one. \\

	Our main tool for the cube state case as well as for the more general case,
	s the following lemma which shows that any small connected $\gamma-separating$ set 
	of edges lies within a small distance from either $e$ or $f$.
	
	\begin{lemma}\label{lem:blockage} {\bf Small connected $\gamma$-seperators 
			are close to end points}
		Let $T_n=(V_n,E_n,F_n)$ be n-th cube CSC. Let $X_n\subseteq E_n$ be copath connected and $\gamma-separating$ and suppose $|X_n|<n$, then $Xn\cap\left(\widehat{Ball}(e,|X_n|)\cup \widehat{Ball}(f,|Xn|)\right)\neq\emptyset$.
	\end{lemma}

	\begin{proof}[Proof of lemma \ref{lem:blockage}]
		Let $x_{0}\in X_n$. Since $X_n$ is copath connected, for
		all $x\in X_n$, there exists a copath $\delta\in\widehat{\Gamma}_{x_{0},x}$such
		that $\delta\subseteq X_n$. Thus, $\widehat{d}(x_{0},x)\leq|\delta|-1 < |X_n|$
		and therefore, $X_n\subseteq \widehat{Ball}(x_{0},|X_n|)$.
		\\
		Assume by contradiction that $X_n\cap\left(\widehat{Ball}(e,|X_n|)\cup \widehat{Ball}(f,|X_n|)\right)=\emptyset$.
		Hence, $\widehat{d}(e,x_{0}),\widehat{d}(f,x_{0}) > |X_n|$ which implies
		$e,f\notin \widehat{Ball}(x_{0},|X_n|)$. 
		
		We now need a claim, which will be very simple to prove
		in the cube case: 
		\begin{claim}\label{cl:pathConnectedExcision}
			Let $e\in E_n$ and $r<n$. Then $T_n \backslash \widehat{Ball}(e,r)$ is path connected.
		\end{claim}
		\begin{proof}{\bf Sketch} This claim is trivial for the case of the cube we are now handling;
			we do not provide the details, since the more general case will be proven later in full. 
		\end{proof}
		
		Hence, we have that $T_n\backslash \widehat{Ball}(x_{0},|X_n|)$ 
		is path connected. Since $e,f\notin \widehat{Ball}(x_{0},|X_n|)$, 
		there exists a path $\gamma'\in\Gamma_{e,f}$ 
		such that $\gamma'\cap \widehat{Ball}(x_{0},|X_n|)=\emptyset$. Hence, 
		$\gamma'\cap X_n=\emptyset$. But since the cube has a trivial first homology group, all paths in $\Gamma_{e,f}$ are in the same homology class, and $\gamma'\in[\gamma]$ contradicting the assumption that $X_n$ is $\gamma-separating$. 
		
	\end{proof}
	
	We can now prove Claim \ref{cl:supportinballs}.
	To do this, recall that we showed in claim \ref{cl:LCaGammaSep}
	that the upper light cone of any element in $A$ satisfies the conditions 
	of the lemma \ref{lem:blockage}, namely it (or a subset of it) 
	is both $\gamma$-separating and co-path connected. 
	Hence we can apply Lemma \ref{lem:blockage} for a subset of 
	the upper light cone 
	of {\it any} element in $A$, to deduce that this subset intersects a small 
	ball around one of the end points; this will allow us to deduce that 
	$A$ is contained in two small balls around the end points (Claim \ref{cl:supportinballs}).
	
	\begin{proof}({\bf Of Claim \ref{cl:supportinballs}}) 
		Assume there exists some edge $g\in A$ such 
		that $g\notin L^{\downarrow}(\widehat{Ball}(e,c^{d_n})
		\cup \widehat{Ball}(f,c^{d_n}))$.
		From claim \ref{cl:LCaGammaSep}, 
		$L^{\uparrow}(g)$ is $\gamma-separating$ and from
		Corollary \ref{cor:AisConnected}
		we can conclude that $L^{\uparrow}(g)$ contains a set $X_n$ which is $\gamma$-separating and copath connected.
		Since each layer in $U$ increases the size of the upper light cone of $g$ 
		by a multiplicative factor of at most $c$, we have $|X_n|\le |L^{\uparrow}(g)|<c^{d_n}$. 
		Therefore, from Lemma \ref{lem:blockage}  we get: 
		\[X_n\cap\left(\widehat{Ball}(e,c^{d_n})\cup \widehat{Ball}(f,c^{d_n})\right)\neq\emptyset\]
		which implies 
		\[L^{\uparrow}(g)\cap\left(\widehat{Ball}(e,c^{d_n})\cup \widehat{Ball}(f,c^{d_n})\right)\neq\emptyset
		\]
		or equivalently, $g\in L^{\downarrow}(\widehat{Ball}(e,c^{d_n})\cup \widehat{Ball}(f,c^{d_n}))$
		and we get a contradiction. 
	\end{proof}

	\subsection{Deducing the Theorem using $\gamma$-separations}
	%In order to show that both $A$ and $B$ are small, 
	%and derive a contradiction with Lemma \ref{}\dnote{add ref}
	%%we will use the following notion: =
	%Observe that whenever a CSC 
	%is simply connected (i.e it has a trivial first homology group), we have $[\gamma]=\Gamma_{e,f}$. Indeed, otherwise, there would exist some $\gamma'\in \Gamma_{e,f}$ such that $\gamma'\notin [\gamma]$. But then, $\gamma + \gamma'$ would be a non trivial cocycle contradicting the assumption of simple connectedness.
	%
	%\dnote{can you explain more how you use the above observation?
	%  I wasn't sure about what is the motivation for this remark here} 
	
	We now use the above Claim \ref{cl:supportinballs} upper bounding $|A|$,
	together with the fact that $A$ must be large (Corollary \ref{cor:Asize}) 
	to derive a contradiction.
	
	We first prove a simple fact bounding the number of edges 
	in a ball on the cube  $T_i$:
	
	\begin{fact}
		\label{fact:area}
		Let $e$ be an edge of the cube CSC $T_i$, and let $d\geq 1$. Then, 
		\[
		|\widehat{Ball}(e,d)|\le 10d^2. 
		\] 
	\end{fact}
	
	\begin{proof}
		Obsesrve that the size of a radius $d$ ball on the cube $T_i$ can always be bounded from above by the size of a ball of the same radius on the infinite grid $\mathbb{Z}\times \mathbb{Z}$. But it 
		is also clear that such a ball is contained inside a $2d\times 2d$ square on the grid, which contains less than $10d^2$ edges. 
	\end{proof}
	
	We are now ready to prove Theorem \ref{thm:nonTrivCubeState}
	
	\begin{proof} [Proof of Theorem \ref{thm:nonTrivCubeState}]
		Assume that $e$ and $f$ are two edges on opposite sides  in $T_n$ 
		so that $d(e,f)\geq n$. We now seperate the set of natural numbers $n$ 
		to two sets. In the first set, we have 
		$n < 2c^{d_n}$ (and so $d_n> log(n)/log(2c)$).
		For the other values of $n$, we have  $d(e,f)\geq  n \ge 2c^{d_n}$. 
		We will derive a contradiction if there are infinitely many  
		$n$'s of the latter type. For any such $n$ 
		we have, applying Claim \ref{cl:supportinballs} and Fact \ref{fact:area}:
		\begin{eqnarray}
		|A| & \leq & |L^{\downarrow}(\widehat{Ball}(e,c^{d_n})\cup \widehat{Ball}(f,c^{d_n})))|\\
		& \leq & c^{d_n}|\widehat{Ball}(e,c^{d_n})\cup \widehat{Ball}(f,c^{d_n}))|\\
		& = & c^{d_n}\left(|\widehat{Ball}(e,c^{d_n})|+|\widehat{Ball}(f,c^{d_n})|\right)\\
		& \leq & 2c^{d_n}\cdot 10c^{2d_n}\\
		& = & 20c^{3d_n}. 
		\end{eqnarray}
		Now, from Corollary \ref{cor:Asize}, 
		we have $|A|> \frac{d(e,f)}{c^d_n} \geq \frac{i}{c^d_n}$. 
		Combining those equations together we get:
		
		\[
		\frac{n}{c^d_n} < 20c^{3d_n})
		\]
		
		from which we get $d_n \geq log(n/20)/4log(c)$ for those $d_n$'s in the second 
		set. Altogether, $d_n$ for all $n$ is greater than the minimum of 
		$log(n/20)/4log(c)$ and $log(n)/log(2c)$,
		and hence $d_n=\Omega(log(n))$. 
	\end{proof}
	
	{\bf Remark:} Observe that this proof heavily relies on 
	fact \ref{fact:area} which provides a polynomial 
	(and even quadratic) upper bound on the amount of edges inside a ball on the cube grid. Unfortunately, this result doesn't hold in the general framework 
	of CSCs, and so this proof cannot be carried over to the more general set 
	of complexes we would like to consider.  
	Indeed, in the general case, the best bound that can be proved 
	is {\it exponential} in the radius of the ball. This kind of behaviour 
	is characteristic of hyperbolic surfaces where the area of a disk of radius r is proportional to $sinh(r)$ instead of the usual $r^2$ in flat euclidean manifolds like the sphere or the torus. This fact alone leads to another {\it log} 
	in the lower bound; but we will in face only derive a $logloglog(n)$ lower 
	boundm due to another phenomenon, as we shall see in the next sections.

	\section{General CSCs}\label{sec:CSCstates}
	In this Section we will prove Theorems \ref{thm:NonGTrivSStates} and \ref{thm:NonTrivSStates},
	generalizing Theorems \ref{thm:NonGTrivCubeState} and \ref{thm:nonTrivCubeState} to general CSCs.
	While most of the proofs - both in the geometrical and non geometrical case - translate to the framework of arbitrary CSCs, observe that we made use of the geometry of the cube both in
	Lemma \ref{lem:SausageIntersection} and lemma \ref{lem:blockage}. in both lemmas, we strongly 
	used the fact that removing any ball of small radius from $T_n$ leaves the cube connected (see Claim \ref{cl:pathConnectedExcision}).
	The difficulty in the case of CSCs is how to guarantee the existance of a path
	between the two edges $e$ and $f$ which does not intersect a set of small radius. To this end, we
	introduce the notion of $r-$simple connectedness, and use it to prove Theorem \ref{thm:getAround} below;
	this replaces the analogous simple fact about the cube, stated in Claim \ref{cl:pathConnectedExcision}. 
	
	\subsection{$r$-simple connectedness}
	\label{sec:alphaCont}

	\begin{definition} 
		Let $G=(V,E,F)$ be a CSC and let $e\in E$ and $r>0$. Define:
		\[
		K'(e,r)=(V',E',F'')
		\]
		where 
		\[
		F''=\{f\in F| \forall x\in f, \widehat{d}(x,e)\leq r\}
		\]
		\[
		V'= V_{F'}
		\]
		\[
		E'= E_{F'}
		\]
		Now for each face $f\in F\backslash F''$, add $f$ to $F''$
		if and only if $\partial_2(f)\subseteq \partial_2(F'')$. (Note that this doesn't change the definition of
		$E'$ and $V'$) Call the resulting set of faces $F'$, and define:
		\[
		K(e,r)=(V',E',F')
		\]
		
	\end{definition}
	
	Note that although $K=K(e,r)=(V',E',F')$ is usually not a CSC as it could have a non empty boundary,
	it is always a polygonal complex.
	
	\begin{definition} {\bf r-simple connectedness}
		\label{def:AlphaContr}
		Let $G=(V,E,F)$ be a CSC, and let $r > 0$. $G$ is called $r$-simply connected
		if for every $e\in E$: $H^{1}(K(e,r))=H_{1}(K(e,r))=0$. 
	\end{definition}
	
	In other words, one can think of an $r-simply connected$ complex as one that doesn't contain any "bottleneck" as illustrated in figure \ref{fig:bottleneck}\\
	The main goal of this section is to prove the following theorem, replacing Claim \ref{cl:pathConnectedExcision}
	by a proof which holds for general CSCs: 
	
	\begin{theorem} 
		\label{thm:getAround}
		Let $G$ be a CSC. Let $r\geq 0$, $e,f,x_0\in E$ and $\gamma\in \Gamma_{e,f}$ such that $e,f\notin B=\widehat{Ball}(x_0,r)$, and $K=K(x_0,r+1)$ is simply connected. Then, there exists $\gamma'\in [\gamma]$ such that $\gamma'\cap B = \emptyset$
	\end{theorem}
	
	First we show that the boundary of any such subcomplex K is path connected, provided its first homology group vanishes:
	
	\begin{lemma}
		\label{lem:ConnectedBoundary}
		Let $G=(V,E,F)$ be a CSC, let $e\in E$, $r>0$ and let  $K=K(e,r)=(V',E',F')$. Assume that $K$ is copath connected and $H_1(K)=0$. Then $\partial_2(F')$ is path connected.
	\end{lemma}

	\begin{proof}
		Assume by contradiction that $\partial_2(F')$ is not path connected. Let $S_1$ and $S_2$ be 2 distinct path connected components of $\partial_2(F')$. Observe that $S_1\in Ker(\partial_1)$. Indeed, for every $v\in V$, 
		$|E_v\cap S_1| = 0 mod 2$ since
		\[
		E_v\cap S_1 = \sum_{f\in F':v\in V_f} E_v\cap \partial_2(f)
		\]
		and for any $f\in F'$, $|E_v\cap \partial_2(f)|$ is either 0 or 2, since by definition \ref{def:CSC}, there are no self-edges. 
		Therefore, $S_1\in Ker(\partial_1)$ and since by assumption $H_1(K)=0$, it follows that $S_1\in Im(\partial_2)$ and we can write $S_1=\sum_{f\in H\subseteq F'} \partial_2(f)$ for some subset $H$ of faces in $F'$. Now let $e_1\in S_1$ and $e_2\in S_2$. By assumption, $K$ is copath connected, hence we can find a copath $(e_1=x_1,x_2,...,x_k=e_2)$ where there is a unique face $f_i$ connecting $x_i$ to $x_{i+1}$ for every $1\leq i\leq k-1$. The uniqueness of $f_i$ stems from the fact that the intersection of 2 faces in a polygonal complex is alway a single edge. Since $x_1\in \partial_2(F')$, we must have $f_1\in K$. But observe that $x_2\in f_1$ and $x_2\notin S_1$ (since $x_2$ belongs to 2 faces). It follows that $f_2\in H$: indeed, if $f_1\in H$ and $f_2\notin H$ then $x_2\in S_1$ and we get a contradiction. By following the same inductive argument, we get that $f_i\in H$ for every $1\leq i\leq k-1$. But then it immediatly follows that $x_k=e_2$ is in $S_1$ since $f_{k-1}\in H$ and $e_2\in S_2$ by assumption. Therefore, $S_1$ and $S_2$ are copath connected in $K$ and hence also path connected in $K$, contradicting the assumption that $S_1$ and $S_2$ are distinct connected component $\partial_2(F')$.
	\end{proof}
	
	We are ready to prove Theorem \ref{thm:getAround}:
	
	\begin{proof} [proof of Theorem \ref{thm:getAround}]
		Let $K=K(x_0, r+1)=(V',E',F')$. Define $B=\widehat{Ball}(x_0,r)$ and $C=E'\backslash(\partial_2(F'))$. Note $B\subseteq C$. Indeed, let $x\in B$: since $\widehat{d}(x_0,x)\leq r$, there exist a copath $(x_0=e_1,...,e_k=x)$ of size  $k\leq r+1$. Let $f\in F$ such that $e_{k-1}\in f$ and $e_k\in f$. Obviously, for all other edges $e'\in f$, it also holds that $\widehat{d}(x_0,e')\leq r\leq r+1$ so that $f\in F'$, and $x\in E_{V_{F'}}= E'$. Let $f'\in F$ be the face satisfying $f'\neq f$ and $x\in f'$. For every $y\in f'$, $(x_0=e_1,...,e_k=x,y)$ is a copath of size $k+1\leq r+2$ and therefore, $\widehat{d}(x_0,y)\leq r$. We conclude that $f'\in F'$, and since we already proved that $f\in F'$, it follows that $x\notin \partial_2(F')$.
		Similarly, $x\notin \partial^0(V_{F'})$ since for any vertex $v\in x$, $v\in V_f\subseteq V_{F'}$. \\
		Now let $\gamma=(e=e_1,...,e_k=f)$. We can partition $\gamma$ into $\gamma=\bigcup_{i=1}^s\gamma_i$ where $\gamma_i\cap \gamma_j=\emptyset$ for $i\neq j$, $\gamma_i=(e_{t_{i}},...,e_{t_i+k_i})$, $t_{i+1}=t_i+k_i+1$ in such a way that for every even i, $\gamma_i\subseteq C$ and for every odd i, $\gamma_i\cap C=\emptyset$. Indeed, since $e,f\notin \widehat{Ball}(x_0,r+1)$, we have $\widehat{d}(x_0,e), \widehat{d}(x_0,f)\geq r+2$ and $e,f\notin C$ and therefore $s$ is odd and the partition defined above is well defined.
		Let $2\leq i\leq s-1$ be an even index. Since  $e_{t_{i}-1}\in\gamma_{i-1}$, $e_{t_{i}-1}\notin C$ and on the other hand,  $e_{t_{i}}\in C$. 
		
		Since $G$ is a $CSC$, we can order the set of edges adjacent to the common vertex of $e_{t_{i}-1}$ and $e_{t_{i}}$, and walk from  $e_{t_{i}-1}$ and $e_{t_{i}}$ through connecting faces. Since $e_{t_{i}-1}\notin C$, it is not connected to a face in $F'$ but  $e_{t_{i}}\in C$ is connected to at least one face from $F'$. Therefore, there must be some edge in the ordering between those 2 edges connected to exctally one face of $F'$, i.e is a boundary of $F'$.
		
		Therefore $e_{t_{i}}\cap x_i\neq \emptyset$ for some $x_i\in \partial_2(F')$. Similarly, we get $e_{t_i+k_i}\cap x'_i\neq \emptyset$ for some $x'_i\in \partial_2(F')$ . From lemma \ref{lem:ConnectedBoundary}, $\partial_2(F')$ is path connected. Hence, there exist a path $\theta_i\in \Gamma_{x_i,x_i'}$ such that $\theta_i\subseteq \partial_2(F')$. But $\gamma_i+\theta_i\subseteq E'$ is a closed path in $K$ and therefore, is a 1-cycle. By assumption, K is simply connected so $H_1(K)=0$ and $\gamma_i+\theta_i=\partial_2(J_i)$ for some subset $J_i\subseteq F'$.
		Now define:
		\[
		\gamma'= \sum_{i=1}^{\frac{s+1}{2}} \gamma_{2i-1} + \sum_{i=1}^{\frac{s-1}{2}} \theta_{2i}
		\] 
		We get:
		\begin{eqnarray}
		\gamma+\gamma' & = &\sum_{i=1}^{s} \gamma_i + \sum_{i=1}^{\frac{s+1}{2}} \gamma_{2i-1} + \sum_{i=1}^{\frac{s-1}{2}} \theta_{2i}\\
		& = &\sum_{i=1}^{\frac{s-1}{2}} \gamma_{2i} + \sum_{i=1}^{\frac{s-1}{2}} \theta_{2i}\\
		& = & \sum_{i=1}^{\frac{s-1}{2}} \gamma_{2i} + \theta_{2i}\\
		& = & \sum_{i=1}^{\frac{s-1}{2}} \partial_2(J_i)\\
		& = &  \partial_2(\sum_{i=1}^{\frac{s-1}{2}} J_i)
		\end{eqnarray}
		and we conclude that $\gamma'\in[\gamma]$.\\
		Finally, for every odd i, $\gamma_i\cap C=\emptyset$ and since $B\subseteq C$, we get $\gamma_i\cap B=\emptyset$. On the other hand, for every even i, $\theta_{i}\subset\partial_2(F')$ and since $C\cap \partial_2(F')=\emptyset$, we also get $\theta_i\cap B=\emptyset$. Therefore, 
		\[
		\gamma'\cap B = \sum_{i=1}^{\frac{s+1}{2}} (\gamma_{2i-1} \cap B) + \sum_{i=1}^{\frac{s-1}{2}} (\theta_{2i} \cap B)=\emptyset
		\]
		This concludes the proof of Theorem \ref{thm:getAround}
	\end{proof}

	\subsection{The geometrically local case}
	We will now prove Theorem \ref{thm:NonGTrivSStates} providing a lower bound on the circuit depth
	when the circuit is {\it geometrically restricted}, when the underlying CSC is $o(log(n))-simply connected$.
	For that purpose we first prove a more general version of Lemma \ref{lem:SausageIntersection}, whose proof
	is essentially almost identical, except for making use of Theorem \ref{thm:getAround}:

	\begin{lemma}
		\label{lem:SausageIntersectionCSC}
		Let G be a CSC and let $U$ be 
		a geometrically-local circuit, whose gates each act on qubits of distance
		at most $c$ apart, and whose depth is at most $d$, satisfying $U\ket{0^n}=\ket{\psi}\in C_G$.
		Let A and B be the lower and upper effective supports with respect to $U$ 
		for some $e,f\in E$ and $\gamma\in \Gamma_{e,f}$. Assume G is $c\cdot d$-simply connected. Then for all n, $A\subseteq \widehat{Ball}(e, cd)\cup \widehat{Ball}(f, cd)$ and $B\subseteq \widehat{Ball}(e, 2cd)\cup \widehat{Ball}(f, 2cd)$
	\end{lemma}
	
	\begin{proof}
		
		We start in the same fashion as in the proof of lemma \ref{lem:SausageIntersection}.
		Let $\gamma'\in [\gamma]$, and we have 
		
		\begin{equation}\label{eq:coneCSC}
		L^{\downarrow}(\gamma')\subseteq \{e'\in E|d(e',\gamma')\leq c\cdot d\}
		\end{equation}

		%Indeed, starting from the support $S$ of $\gamma'$ and applying only the first layer from the reversed circuit $U^\dagger$, the lower light cone of $\gamma'$ can only extend its reach to qubits within a distance $c-1$ from $S$. Indeed, since the quantum gates composing $U$ are spatially local, they can only affect qubits within a distance $c-1$ from each other. Applying this bound recursively on D layers proves (3).
		Let $e'\in E$ and assume that
		$e'\notin A\subseteq \widehat{Ball}(e, cd)\cup \widehat{Ball}(f, cd),$
		namely         that $d(e,e')>cd$ and $d(f,e')>cd$. 
		Observe that $\widehat{ball}(e',cd)$ doesn't contain neither $e$ or $f$. 
		Therefore, from the assumption of $c\cdot d$-simple connectedness and from Theorem \ref{thm:getAround},
		there exists some path $\gamma''\in [\gamma]$ such that $d(\gamma'',e')>cd$ and then from Equation \ref{eq:coneCSC}, 
		$e'\notin L^{\downarrow}(\gamma'')$, and hence it is not in $A$. 
		This proves the first part of the claim.
		Since $B= L^\uparrow(A)$, by definition any qubit in $B$ is within distance 
		at most $cd_n$ from any qubit in $A$ and thus $B$ is contained in the set of 
		all qubits at distance at most $cd_n$ from 
		$\widehat{Ball}(e, cd_n)\cup \widehat{Ball}(f, cd_n)$; this set is contained in 
		$\widehat{Ball}(e, 2cd_n)\cup \widehat{Ball}(f, 2cd_n)$ and the second part of the claim follows.
	\end{proof}
	
	As we stated in the last section, Fact \ref{fact:area}
	doesn't hold in the most general non-euclidean case, as can be seen in figure \ref{fig:hyperbolic}. What we have instead is:
	\begin{lemma}
		\label{lem:EdgesInBall}
		Let G be a simple graph with n edges and bounded degree $deg(G)\leq D$. Then for every edge e in G, $|Ball(e,n)|\leq D^{n+1}$
	\end{lemma}
	
	\begin{proof}
		Define $C_n=\{g\in E|d(e,g)=n\}$:\\
		Observe that $|C_n|\leq 2D^n$ from a simple union bound.
		%	For n=1, each of the 2 end vertices of e belongs to at most $D$ edges and share at most one common edge: e itself. therefore, $C_1=2(D-1)\leq 2D$\\
		%	Now assume that $C_i\leq 2D^i$. Let $g\in C_{i+1}$ and let  $\gamma=(e_1=e,...,e_{i+1},e_{i+2}=g)\in \Gamma_{e,g}$ be a path of minimal length n+1. Obviously, $\gamma \backslash\{g\}=(e_1=e,...,e_{i+1})$ is a minimal path between $e$ and $e_{i+1}$ and therefore $e_{i+1}\in C_i$. We conclude that  $|C_{i+1}|\leq D\cdot |C_i|\leq D\cdot 2D^i=2D^{i+1}$.\\
		%	Now since $Ball(e,n)=\cup_{i=0}^n C_i$ and the union is disjoint, 
		Therefore we get:
		\begin{eqnarray}
		|Ball(e,n)| & = & |\cup_{i=0}^n C_i|\\
		& = & 1 + \sum_{i=1}^n |C_i|\\
		& \leq & 1+ \sum_{i=1}^n 2D^i\\
		& \leq & 1+ 2D\cdot \frac{D^{n}-1}{D-1}\\
		& \leq &  1+ D\cdot (D^{n}-1)\\
		& \leq &  D^{n+1}\\
		\end{eqnarray}
	\end{proof}
	
	\begin{figure}
		\begin{center}
			\includegraphics[scale=0.3]{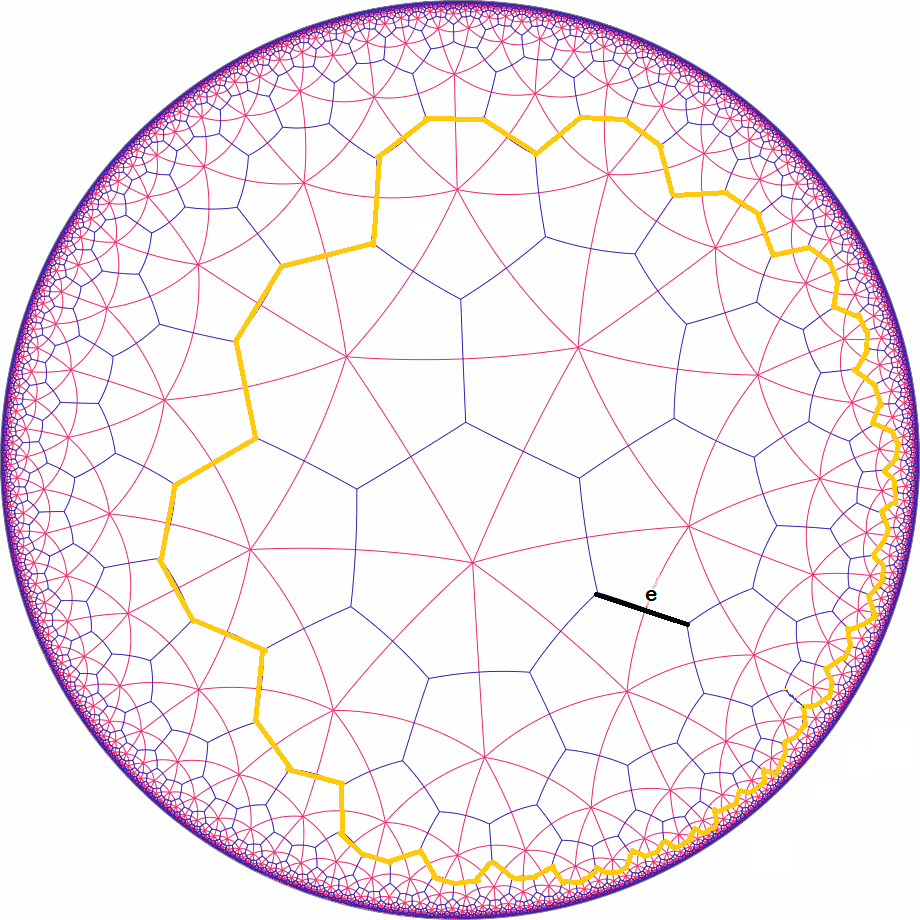} 
		\end{center}
		\caption{\label{fig:hyperbolic} The yellow path represents the circle of radius r=3 around the edge e. In hyperbolic geometry, both the circumference of a circle and the area it defines grow exponentially with the radius.}
	\end{figure}
	
	\begin{corollary}
		\label{cor:diamBound}
		Let G be a simple graph with n edges and bounded degree $deg(G)\leq D$. Then $diam(G)\geq log_D(n)-1$
	\end{corollary}
	
	\begin{proof}
		Let  $d=diam(G)$ and let $e\in E$. We have $E=Ball(e,d)$ and from the previous lemma, we get $|E|=n\leq D^{d+1}$. The corollary follows.
	\end{proof}
	
	We now give a formal proof of theorem \ref{thm:NonGTrivSStates}:
	
	\begin{proof} [proof of Theorem \ref{thm:NonGTrivSStates}]
		For each n, choose 2 edges $e,f$ verifying $d(e,f)=diam(G_i)$, and choose some path $\gamma\in\Gamma_{e,f}$. From corollary \ref{cor:diamBound} we have $d(e,f)>log(N_i)$. Now assume by contradiction that there exists a depth $d_i\neq \Omega(log(N_i))$ geometrically local quantum circuit $U_i$ such that $U_i\ket{0}^N_i = \ket{\psi_i}$. We can extract a subsequence $d_{i_j}=o(log(N_{i_j}))$. From lemma \ref{lem:SausageIntersectionCSC}, since $G_{i_j}$ is $O(log(N_{i_j}))$-simply connected, we get $B=B_{[\gamma],U_{i_j}}\subseteq \widehat{Ball}(e, 2cd_{i_j})\cup \widehat{Ball}(f, 2cd_{i_j})$. But for large enough j, $\widehat{Ball}(e, 2cd_{i_j})\cap \widehat{Ball}(f, 2cd_{i_j})=\emptyset$, and therefore, B doesn't contain any path in $[\gamma]$ contradicting lemma \ref{lem:LargeB}.
	\end{proof}
	
	\subsection{The non geometrically local case}
	Here we prove our main result, Theorem \ref{thm:NonTrivSStates}, providing a lower bound on circuit depth
	for the non-geometrical case, for general CSCs with sufficient $r$-simple connectedness. 
	We first prove a modified version of Lemma \ref{lem:blockage} which generalizes it to any $r$-simply connected CSCs for
	suitable $r$.

	\begin{lemma}
		\label{lem:contract}
		Let $G=(V,E,F)$ be a CSC. Let $X\subseteq E$ be copath connected and $\gamma$-separating.
		Define $r=|X|$ and assume $G$ is $r+1$-simply connected, then
		$X\cap\left(\widehat{Ball}(e,r)\cup \widehat{Ball}(f,r)\right)\neq\emptyset$.
	\end{lemma}
	
	\begin{proof}
		Let $x_{0}\in X$. Since $X$ is copath connected, for
		all $x\in X$, there exists a copath $\widehat{\gamma}\in\widehat{\Gamma}_{x_{0},x}$ such
		that $\widehat{\gamma}\subseteq X$. Thus, $\widehat{d}(x_{0},x)\leq|\widehat{\gamma}|-1 \leq r$
		and therefore, $X\subseteq  \widehat{Ball}(x_{0},r)$.
		Assume by contradiction that $X\cap\left(\widehat{Ball}(e,r)\cup \widehat{Ball}(f,r)\right)=\emptyset$.
		Hence, $\widehat{d}(e,x_{0}),\widehat{d}(f,x_{0})\geq r$ and
		so $e,f\notin \widehat{Ball}(x_{0},r)$.
		We can now use Theorem \ref{thm:getAround} to find a path $\gamma'$
		that satisfies $\gamma'\in[\gamma]$ and $\gamma'\cap \widehat{Ball}(x_{0},r)=\emptyset$.
		Since $X\subseteq \widehat{Ball}(x_{0},r)$ we also get 
		$\gamma'\cap X=\emptyset$ and that's
		a contradiction to the assumption that $X$ is $\gamma-separating$.
	\end{proof}

	Finally we turn to the proof of our main result:
	
	\begin{proof} [Proof of Theorem \ref{thm:NonTrivSStates}]
		This proof will be almost identical to the proof of Theorem \ref{thm:nonTrivCubeState}:
		Assume that $e$ and $f$ are 2 edges in $E_i$ such that $d(e,f)=diam(G_i)$.
		We first prove that $A\subseteq L^{\downarrow}(\widehat{Ball}(e,c^{d_i})\cup \widehat{Ball}(f,c^{d_i}))$.
		Indeed, assume by contradiction that
		there exists some edge $g\in A$ such that $g\notin L^{\downarrow}(\widehat{Ball}(e,c^{d_i})\cup \widehat{Ball}(f,c^{d_i}))$.
		From claim \ref{cl:LCaGammaSep}, 
		$L^{\uparrow}(g)$ is $\gamma-separating$ and from
		corollary \ref{cor:AisConnected}
		we may assume $L^{\uparrow}(g)$ is copath connected.\\
		Since each layer in U increases the size of the upper light cone of g by a multiplicative factor of at most c, we have $|L^{\uparrow}(g)|<c^{d_i}$. By assumption, $G_i$ is $f(N_i)$-simply connected for all $f=o(log(log(N_i)))$:
		
		Now if $d_i= \Omega(log(log(log(N_i))))$, we are done. 
		Assume otherwise. Then we can extract a subsequence $d_{i_j}=o(log(log(log(N_{i_j}))))$ and $|L^{\uparrow}(g)| < c^{d_{i_j}} = o(log(log(N_{i_j})))$. In this case, we can apply
		lemma \ref{lem:contract} to get: 
		\[
		L^{\uparrow}(g)\cap\left(\widehat{Ball}(e,c^{d_{i_j}})\cup \widehat{Ball}(f,c^{d_{i_j}})\right)\neq\emptyset
		\]
		Hence, $g\in L^{\downarrow}(L^{\uparrow}(g)\cap\left(\widehat{Ball}(e,c^{d_{i_j}})\cup \widehat{Ball}(f,c^{d_{i_j}})\right))\subseteq L^{\downarrow}(\widehat{Ball}(e,c^{d_{i_j}})\cup \widehat{Ball}(f,c^{d_{i_j}}))$
		and we get a contradiction.
		We deduce that $A\subseteq L^{\downarrow}(\widehat{Ball}(e,c^{d_{i_j}})\cup \widehat{Ball}(f,c^{d_{i_j}}))$.\\
		Let's bound the size of A from above:
		\begin{eqnarray}
		|A| & \leq & |L^{\downarrow}(\widehat{Ball}(e,c^{d_{i_j}})\cup \widehat{Ball}(f,c^{d_{i_j}})))|\\
		& \leq & c^{d_{i_j}}|\widehat{Ball}(e,c^{d_{i_j}})\cup \widehat{Ball}(f,c^{d_{i_j}}))|\\
		& \leq & c^{d_{i_j}}\left(|\widehat{Ball}(e,c^{d_{i_j}})|+|\widehat{Ball}(f,c^{d_{i_j}})|\right)\\
		& \leq & 2c^{d_{i_j}}\cdot D^{c^{d_{i_j}}+1}
		\end{eqnarray}
		where we used lemma \ref{lem:EdgesInBall} in the last inequality, and $D=max(deg(G),\widehat{deg}(G))$ (see Definition \ref{def:VertFaceDeg}) 
		Now, from Corollary \ref{cor:Asize}, we have $|A|> \frac{d(e,f)}{c^{d_{i_j}}} = \frac{diam(G)}{c^{d_{i_j}}}$ and from corollary \ref{cor:diamBound}, we get $|A| > \frac{log_D(N_{i_j})-1}{c^d_{i_j}}$. Combining those equations together we get:
		
		\[
		\frac{log_D(N_{i_j})-1}{c^{d_{i_j}}} < 2c^{d_{i_j}}\cdot D^{c^{d_{i_j}}+1}
		\]
		or:
		\[
		log_D(N_{i_j}) < 2c^{2d_{i_j}}\cdot D^{c^{d_{i_j}}+1}+1
		\]
		
		and it follows that $d_{i_j} \neq o(log_c(log_D(log_D(|N_{i_j}|))))$ leading to a contradiction.
	\end{proof}

	\section{Appendix A: trivial states are UGS of gapped Hamiltonians}
	\label{appA}
	
	Assume $\{\ket{\psi_n}\}$ is a familly of trivial states on n qubits. Then we can find a family of constant depth circuits $\{U_n\}$ such that $U_n\ket{0^n}=\ket{\psi_n}$. Let $H_n^i=U_n Z^i U_n^\dagger$ where $Z^i\triangleq z^i \otimes I_{[n]\backslash\{i\}}$ and $z^i$ is the z Pauli operator on the i-th qubit. Now, define the following Hamiltonian:
	
	\[
	K_n = - \sum_{i=1}^{n} H_n^i
	\]
	
	We will first show that $H_n^i$ act on a constant number of qubits, namely:
	
	\begin{lemma}
		\label{lem:lightConeExtension}
		Let U be a depth d quantum circuit acting on n qubits and let $P$ be an operator such that $|supp(P)|=k$. Then $|supp(U P U^\dagger)|\leq c^d k$ where $c$ is the maximal size of the support of any gate in $U$.
	\end{lemma}
	
	\begin{proof}
		Assume $U=U_d\cdot U_{d-1}\cdot ...\cdot U_1$ where each layer $U_i$ is a tensor product of unitaries $U_i=V_i^1\otimes ...\otimes V_i^{t_i}$ where each $V_i^j$ represents a quantum gate acting on a subset $S_{i,j}$ of size at most c and for $j\neq j'$, $S_{i,j}\cap S_{i,j'}=\emptyset$. 
		For $0\leq i\leq d$, define $W_i=U_i\cdot U_{i-1}\cdot ...\cdot U_1$ (assuming $W_0=I$). We now prove the lemma by showing by induction that the operators $P_i = W_i P W_i^\dagger$ are supported on $c^i k$ qubits.\\
		For $i=0$, we have $P_0 = W_0 P W_0^\dagger=P$, and by definition, $|Supp(P)|=k=c^0 k$.\\
		Now assume $Supp(P_i)=K_i$ where $|K_i|=c^i k$. Observe that $P_{i+1} =U_{i+1} P_i U_{i+1}^\dagger$. 
		Define 
		\[
		K_{i+1} = \bigcup_{S_{i+1,j}\cap K_i \neq \emptyset} S_{i+1,j}
		\] 
		We argue that $Supp(P_{i+1})=K_{i+1}$. Indeed, We can write:
		\[
		U_{i+1}=\bigotimes_{S{i+1,j}\cap K_i \neq \emptyset} V_{i+1}^j \bigotimes_{S{i+1,j}\cap K_i = \emptyset} V_{i+1}^j
		\] 
		Define $U_{i+1}^{(0)} = \bigotimes_{S{i+1,j}\cap K_i \neq \emptyset} V_{i+1}^j$ and $U_{i+1}^{(1)} = \bigotimes_{S{i+1,j}\cap K_i = \emptyset} V_{i+1}^j$.
		We can write $P_i = \tilde{P_i} \otimes I_{[n]\backslash K_i}$ where $\tilde{P_i}$ acts on qubits in $K_i$ only. Therefore:
		\begin{eqnarray}
		P_{i+1} & = & U_{i+1} P_i U_{i+1}^\dagger \\
		& = & (U_{i+1}^{(0)} \otimes U_{i+1}^{(1)}) (\tilde{P_i} \otimes I_{[n]\backslash K_i}) (U_{i+1}^{(1)\dagger} \otimes U_{i+1}^{(0)\dagger})\\
		& = & (U_{i+1}^{(0)} \tilde{P_i} U_{i+1}^{(0)\dagger}) \otimes (U_{i+1}^{(1)} I_{[n]\backslash K_i}  U_{i+1}^{(0)\dagger})\\
		& = & (U_{i+1}^{(0)} \tilde{P_i} U_{i+1}^{(0)\dagger}) \otimes I_{[n]\backslash K_{i+1}}
		\end{eqnarray}
		
		This shows that indeed, $Supp(P_{i+1})=K_{i+1}$. Finally we are left to show that $|K_{i+1}|\leq c\cdot |K_i|$. Indeed, observe that:
		\begin{eqnarray}
		|K_{i+1}| & = & |\bigcup_{S_{i+1,j}\cap K_i \neq \emptyset} S_{i+1,j}|\\
		& = & \sum_{S_{i+1,j}\cap K_i \neq \emptyset} |S_{i+1,j}|\\
		& \leq & \sum_{l\in K_i,l\in S_{i+1,j}} |S_{i+1,j}|\\
		& \leq & |K_i|\cdot max_j |S_{i+1,j}|\\
		& \leq & c\cdot |K_i|
		\end{eqnarray}
		
	\end{proof}
	
	We now turn to the proof of the main theorem:
	
	\begin{theorem}
		$\{K_n\}$ is a family of non degenerate gapped comuting local Hamiltonians. Furthermore, for all $n\geq 1$, $\ket{\psi_n}$ is the unique ground state of $K_n$
	\end{theorem}
	
	\begin{proof}
		Let $E_n= \braket{\psi_n}{K_n|\psi_n} $.
		For all $1\leq i\leq n$, $||Z_i||=1$ (the operator norm of $Z_i$), so $||H_n^i||=1$ since conjugation by a unitary operator preserves the norm. Therefore,
		\[
		||K_n||\leq \sum_{i=1}^{n} ||H_n^i|| = n
		\]. Observe that
		
		\begin{eqnarray}
		E_n & = & -\braket{\psi_n}{\sum_{i=1}^{n} H_n^i|\psi_n} \\
		& = & -\sum_{i=1}^{n} \braket{\psi_n}{H_n^i|\psi_n}\\
		& = & -\sum_{i=1}^{n} \braket{\psi_n}{U^i Z^i U^{i\dagger}|\psi_n}\\
		& = &  -\sum_{i=1}^{n} \braket{0^n}{Z_i| 0^n}\\
		& = & -n
		\end{eqnarray}
		
		Therefore, $\ket{\psi_n}$ is in the ground space of $K_n$.  \\
		On the other hand assume that another state $\ket{\psi_n'}$ also had energy $E_n'=\braket{\psi_n'}{K_n|\psi_n'}=-n$. This can only happen if each of the n terms $H_n^i$ has energy $E_n^i = \braket{\psi_n'}{H_n^i| \psi_n'}=1$ or equivalently, $\braket{\phi_n}{Z_i| \phi_n}=1$ where we define $\ket{\phi_n}=U_n^\dagger \ket{\psi_n'}$. We conclude that for every $i$, $\ket{\phi_n}$ is a +1 eigenstate of $Z_i$. Therefore,  $\ket{\phi_n}=\ket{0^n}$ and by applying the unitary $U_n$ on both sides we  get $\ket{\psi_n'}=U_n \ket{\phi_n}=U_n \ket{0^n}= \ket{\psi_n}$. \\
		Observe that the terms $\{H_n^i\}_i$ are pairwise commuting since:
		\[
		[H_n^i,H_n^j]=[Z_i,Z_j]=(-1)^{\delta_{i,j}}
		\] 
		We now prove that $\{K_n\}$ is gapped: Indeed, the second energy level of $\{K_n\}$ is met when exactly one local term $H_n^i$ has energy -1. In such a case, 
		`
		\[
		E_n^{(2)} =- [(n-1)\cdot 1 + 1\cdot (-1)]= n-2 
		\]
		and we get a constant spectral gap $\Delta E=n-(n-2)=2$.\\
		
	\end{proof}

	\section{Appendix B: Non-trivial unique groundstates}
	\label{appB}
	
	Here we show a generic construction (proposed to us By I. Arad \cite{arad}) 
	of a non-gapped Hamiltonian with a unique non-trivial groundstate.\\

	We use the quantum verifier to Hamiltonian construction of Kitaev in \cite{KitaevQMA} applied on a family of circuits $\{U_i\}$ generating the i-th CAT state:
	\[
	\ket{CAT^+}=\frac{1}{\sqrt{2}}(\ket{0}^n+\ket{1}^n)
	\]
	The n-th circuit $U_n$ can be described as a product of $n^2$ layers of unitaries $U_n=V_n^{n^2}\cdot ...\cdot V_n^1$ 
	where we define:
	\[
	\forall i,1\leq i\leq n: \: V_n^i=CNOT_{i,i+1}\otimes I_{[n]\backslash\{i,i+1\}}
	\]
	\[
	\forall i,n+1\leq i\leq n^2: \: V_n^i=I_[n]
	\]
	
	Observe that this is a depth $n^2$ circuit that generates the $\ket{CAT^+}$ state on the input $\ket{init}_n=\ket{+}\otimes \ket{0}^{n-1}$
	The actual state we will be looking at is the computation history state defined on $n+log(n^2)$ qubits as:
	\[
	\ket{\varphi_n}=\frac{1}{n}\sum_{i=1}^{n^2} U_n^i \ket{init_n} \otimes \ket{i}
	\]
	
	where $U_n^i=V_n^i\cdot ...\cdot V_n^1$.
	This state is the unique ground state of the $O(1)$-local Hamiltonian:
	\[
	H= H_{out} + J_{in} H_{in} + J_{prop} H_{prop}
	\]
	defined in chapter 4 of \cite{KitaevQMA}. 
	It is straigntforward from Kitaev's construction that $H$ has unique ground state $\ket{\varphi_n}$ (since it is the only state satisfying the local constraints of $H$ all together, describing the unique computation of $U_n$ on $\ket{init}_n$).\\
	Kitaev also proves that the spectral gap of $H$ satisfies $\Delta(H)=O(\frac{1}{n^2})$ so that $H$ is not gapped.\\
	We show that $\ket{\varphi_n}$ is also non trivial: To see that, we look at the extended CAT state: $\ket{\tilde{CAT_n}^+} = \ket{CAT_n^+}\otimes \ket{\theta_{n}}$ where $\ket{\theta_{n}}$ is the state $\ket{\theta_{n}} = \frac{1}{\sqrt{n^2-n}} \sum_{i=n+1}^{n^2} \ket{i}$ on $log(n^2)$ qubits.
	\begin{eqnarray}
	|\braket{\varphi_n}{\tilde{CAT}^+}| & = & |\frac{1}{n} \sum_{i=1}^{n^2} (\bra{init_n} U_n^{i\dagger} \otimes \bra{i}) (\ket{CAT_n^+}\otimes \ket{\theta_{n}})|\\
	& \geq & |\frac{1}{n} \sum_{i=n+1}^{n^2} (\bra{init_n} U_n^{i\dagger} \otimes \bra{i}) (\ket{CAT_n^+}\otimes \ket{\theta_{n}})| - |\frac{1}{n} \sum_{i=1}^{n} (\bra{init_n} U_n^{i\dagger} \otimes \bra{i}) (\ket{CAT_n^+}\otimes \ket{\theta_{n}})|\\
	& = & \frac{1}{n} |\sum_{i=n+1}^{n^2} (\bra{init_n} U_n^{\dagger} \otimes \bra{i}) (\ket{CAT_n^+}\otimes \ket{\theta_{n}})|\\
	& = & \frac{1}{n} |(\bra{init_n} U_n^{\dagger} \otimes \sum_{i=n+1}^{n^2}\bra{i}) (\ket{CAT_n^+}\otimes \ket{\theta_{n}})|\\
	& = & \frac{1}{n} |(\bra{init_n} U_n^{\dagger} \otimes \sqrt{n^2-n} \ket{\theta_n}) (\ket{CAT_n^+}\otimes \ket{\theta_{n}})|\\
	& = & \frac{\sqrt{n^2-n}}{n} |\braket{init_n|U_n^\dagger}{CAT^+_n}|\\
	& = & \frac{\sqrt{n^2-n}}{n}\\
	\end{eqnarray}
	
	And therefore, $lim_{n\rightarrow \infty} \braket{\varphi_n}{\tilde{CAT}^+} = 1$. 
	Since $\tilde{CAT}^+$ exhibit long range correlations between its first $n$ qubits (as it is basically the usual $\ket{CAT^+}$ state on those qubits) so does $\ket{\psi_n}$. Since we know that the output of a constant depth circuit exhibits no such correlations, we conclude that $\ket{\psi_n}$ is non trivial.

	\section{Appendix C: k-UDA implies k-UDP}
	\label{appD}
	Assume the n qubit state $\ket{\psi}$ is the unique ground state of some k-local hamiltonian $H=\sum_i H_i$, and suppose by contradiction that $\rho \neq \ketbra{\psi}{\psi}$ satisfies that for all $K\subseteq [n]$, $|K|\leq k$. $Tr_{\bar{K}}\ketbra{\psi}{\psi} = Tr_{\bar{K}}\rho$ then:
	
	\begin{eqnarray}
	Tr(H_n\rho) & = & Tr(\sum_i H_i\rho)\\
	& = & \sum_i Tr(H_i\rho)\\
	& = & \sum_i Tr(H_i Tr_{\bar{K_i}}\rho)\\
	& = & \sum_i Tr(H_i Tr_{\bar{K_i}}\ketbra{\psi}{\psi})\\
	& = & \sum_i Tr(H_i\ketbra{\psi}{\psi})\\
	& = & Tr(H_n\ketbra{\psi}{\psi})\\
	\end{eqnarray}
	
	Therefore, $\rho \neq \ketbra{\psi}{\psi}$ have the same energy. 
	Now write $\rho = \sum_i p_i\ketbra{\psi_i}{\psi_i}$ where $p_i\geq 0$ and $\sum_i p_i=1$. We have:
	
	\[
	Tr(\rho H)=\sum_i p_i\braket{\psi_i}{H|\psi_i}
	\]
	Since $\rho$ has minimal energy, every element $\braket{\psi_i}{H|\psi_i}$ in the above weighted sum must be minimal and it follows that $\ket{\psi_i}=\ket{\psi}$ for all i. Hence, $\rho = \ketbra{\psi}{\psi}$ and we get a contradiction.

\end{document}